\documentclass[11pt]{article}  
\usepackage[margin=1in]{geometry}
\usepackage[utf8]{inputenc}
\usepackage[T1]{fontenc}
\usepackage{lmodern}
\usepackage{amssymb,amsmath,amsthm,amsfonts}
\usepackage{textgreek}
\usepackage{mathtools}
\usepackage{empheq} 
\usepackage{enumitem}
\usepackage[numbers,comma,sort&compress]{natbib}
\usepackage{authblk}
\usepackage{graphicx, caption, subcaption}
\usepackage{svg}
\usepackage{float}
\usepackage[ruled,vlined,linesnumbered]{algorithm2e}
\usepackage{algpseudocode}
\usepackage{siunitx}
\usepackage{physics}
\usepackage{footnote}
\usepackage{xcolor}
\usepackage{mathrsfs}
\usepackage{bbm}
\usepackage{bm}

\usepackage[colorlinks]{hyperref}

\theoremstyle{plain}
\newtheorem{theorem}{Theorem}
\newtheorem{proposition}[theorem]{Proposition}
\newtheorem{lemma}[theorem]{Lemma}

\theoremstyle{definition}
\newtheorem{definition}{Definition}

\newtheorem{remark}{Remark}

\newcommand{\eqn}[1]{\hyperref[eqn:#1]{(\ref*{eqn:#1})}}
\newcommand{\rem}[1]{\hyperref[rem:#1]{Remark~\ref*{rem:#1}}}
\newcommand{\thm}[1]{\hyperref[thm:#1]{Theorem~\ref*{thm:#1}}}
\newcommand{\cor}[1]{\hyperref[cor:#1]{Corollary~\ref*{cor:#1}}}
\newcommand{\defn}[1]{\hyperref[defn:#1]{Definition~\ref*{defn:#1}}}
\newcommand{\assump}[1]{\hyperref[assump:#1]{Assumption~\ref*{assump:#1}}}
\newcommand{\lem}[1]{\hyperref[lem:#1]{Lemma~\ref*{lem:#1}}}
\newcommand{\prop}[1]{\hyperref[prop:#1]{Proposition~\ref*{prop:#1}}}
\newcommand{\fig}[1]{\hyperref[fig:#1]{Figure~\ref*{fig:#1}}}
\newcommand{\tab}[1]{\hyperref[tab:#1]{Table~\ref*{tab:#1}}}
\newcommand{\algo}[1]{\hyperref[algo:#1]{Algorithm~\ref*{algo:#1}}}
\renewcommand{\sec}[1]{\hyperref[sec:#1]{Section~\ref*{sec:#1}}}
\newcommand{\append}[1]{\hyperref[append:#1]{Appendix~\ref*{append:#1}}}
\newcommand{\fac}[1]{\hyperref[fac:#1]{Fact~\ref*{fac:#1}}}
\newcommand{\lin}[1]{\hyperref[lin:#1]{Line~\ref*{lin:#1}}}
\newcommand{\fnote}[1]{\hyperref[fnote:#1]{Footnote~\ref*{fnote:#1}}}

\def\>{\rangle}
\def\<{\langle}

\newcommand{\vect}[1]{\ensuremath{\boldsymbol{#1}}}

\newcommand{\N}{\mathbb{N}}

\newcommand{\R}{\mathbb{R}}
\newcommand{\C}{\mathbb{C}}

\renewcommand{\H}{\mathcal{H}}

\renewcommand{\d}{\mathrm{d}}

\newcommand{\xx}{\vect{x}}
\newcommand{\D}{\mathcal{D}}
\newcommand{\bigO}{\mathcal{O}}
\newcommand{\tbigO}{\widetilde{\mathcal{O}}}
\newcommand{\ttspar}[1]{T_{#1}}

\DeclareMathOperator{\poly}{poly}

\renewcommand\bra[1]{{\langle{#1}|}}
\makeatletter
\renewcommand\ket[1]{%
  \@ifnextchar\bra{\k@t{#1}\!}{\k@t{#1}}%
}
\newcommand\k@t[1]{{|{#1}\rangle}}
\makeatother

\numberwithin{equation}{section}

\begin{document}

\title{A Quantum-Classical Performance Separation in Nonconvex Optimization}
\author[1,3] {Jiaqi Leng}
\author[2,3] {Yufan Zheng}
\author[2,3,$\dagger$] {Xiaodi Wu}
\affil[1]{Department of Mathematics, University of Maryland}
\affil[2]{Department of Computer Science, University of Maryland}
\affil[3]{Joint Center for Quantum Information and Computer Science, University of Maryland}
\affil[$\dagger$]{\href{mailto:xiaodiwu@umd.edu}{xiaodiwu@umd.edu}}
\date{}

\maketitle

\begin{abstract}
    In this paper, we identify a family of nonconvex continuous optimization instances, each $d$-dimensional instance with $2^d$ local minima, to demonstrate a quantum-classical performance separation. 
    Specifically, we prove that the recently proposed Quantum Hamiltonian Descent (QHD) algorithm [Leng et al., arXiv:2303.01471] is able to solve any $d$-dimensional instance from this family using $\tbigO(d^3)$\footnote{$\tbigO(\cdot)$ suppresses poly-logarithmic factors in $d$ and $1/\delta$ where $\delta$ is the precision.} quantum queries to the function value and $\tbigO(d^4)$ additional 1-qubit and 2-qubit elementary quantum gates. 
    On the other side, a comprehensive empirical study suggests that representative state-of-the-art classical optimization algorithms/solvers (including Gurobi) would require a super-polynomial time to solve such optimization instances.
\end{abstract}

\section{Introduction}
Nonconvex optimization is a central object in machine learning and operations research. In practice, iterative optimization algorithms using local gradient information (e.g., stochastic gradient descent) have been proven successful for large-scale nonconvex optimization, especially when the optimization landscape is well-conditioned. However, lots of nonconvex problems naturally arising from application domains possess sophisticated optimization landscape with a huge number of saddle points and spurious local minima, posing great challenges for gradient-based algorithms. Heuristic optimization algorithms (e.g., simulated annealing, particle swarm, etc.), on the other hand, could explore the nonconvex landscape more efficiently, while few theoretical guarantees are established and the performance of heuristic algorithms varies drastically between instances.

Quantum computers are revolutionary computing machines that leverage the law of quantum mechanics to deliver faster and more secure resolutions to involved computational tasks. Since quantum computers directly operate quantum states that represent the superposition of all possible solutions, they can potentially speed up the exploration of the optimization landscape and therefore demonstrate real advantage for nonconvex optimization problems. 

Recently, Leng, Hickman, Li, and Wu~\cite{leng2023quantum} propose a quantum algorithm named \textbf{Quantum Hamiltonian Descent (QHD)} for continuous optimization.
QHD can be regarded as a quantum-upgraded version of classical gradient descent because of its simplicity and resource efficiency; on the other hand, QHD goes beyond local search as the quantum tunneling effect enables QHD to go through barriers on the optimization landscape. Although numerical and real-machine experiments suggest that QHD has potent performance for continuous optimization, a fine-grained theoretical understanding of QHD's speedup in the nonconvex setting remains an open question. 

In this work, we further investigate of behavior of QHD on nonconvex continuous optimization problems. We construct a specific family of optimization instances such that: (1) the number of local minima on these problems scales exponentially with the problem dimension; (2) no explicit sparsity or separability pattern can be employed to enable an efficient classical solution.
Due to these features, the constructed problem instances are intractable for most classical optimization algorithms. Empirically, we find that several state-of-the-art classical optimization algorithms and solvers would require a \emph{super-polynomial} time to solve these hard instances (see \sec{empirical}). Meanwhile, we prove that QHD solves these hard optimization problems in \emph{polynomial} time. Our findings constitute new evidence of a significant performance separation between quantum and classical algorithms in nonconvex optimization.

\subsection{Problem formulation}
We construct a family of optimization instances representing nonconvex optimization problems with sophisticated landscapes. These hard problem instances are constructed as follows.
First, we consider certain well-formed one-dimensional ``double well'' functions with two local minima (see \fig{instances}A). These well-formed double well functions can be expressed as degree-4 polynomials (see \sec{instances} for more details). Suppose that $w(x)$ is a well-formed 1D double well function with global minimizer at $x = x^*$, we then define a $d$-dimensional separable function $F(\xx) \coloneqq \sum^d_{k=1} w(x_k)$. We note that the function $F(\xx)$ has $2^d$ local minima and a unique global minimizer at $\xx^* = (x^*,\dots,x^*)$. In \fig{instances}B, we show a two-dimensional separable function with 4 local minima.

\begin{figure}[htbp]
    \centering
    \includegraphics[width=12cm]{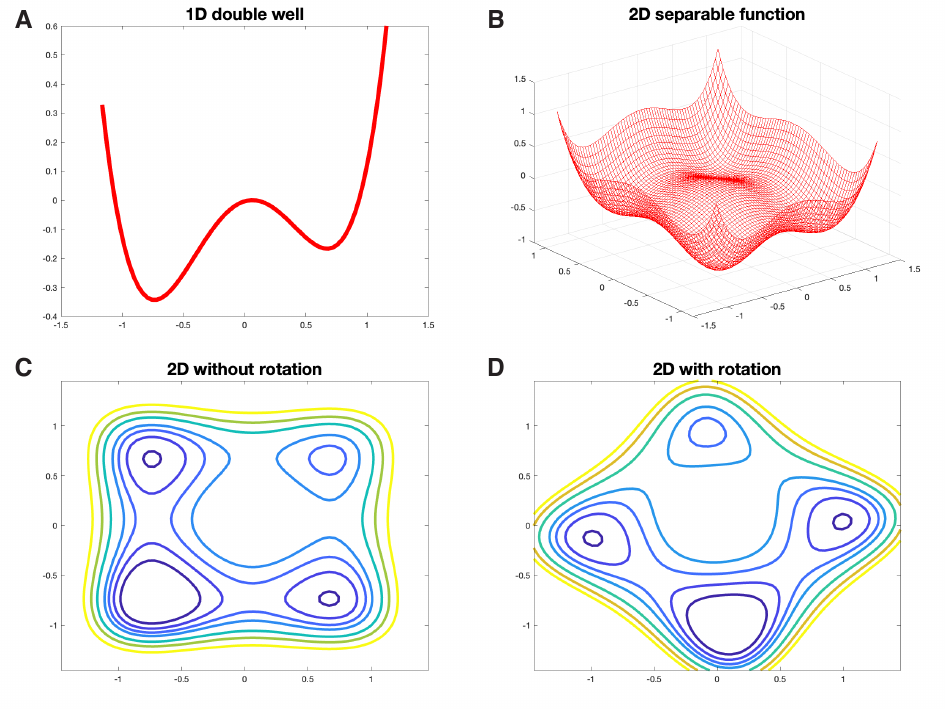}
    \caption{\small \textbf{Construction of optimization instances.} 
    \textbf{A.} A well-formed double well function.
    \textbf{B.} A separable function built from two identical well-formed double well functions. 
    \textbf{C.} Contour plot of the 2-dimensional separable function $F(\xx)$.
    \textbf{D.} Contour plot of the rotated 2-dimensional function $F_U(\xx)$. The local geometry of $F(\xx)$ is preserved, while the new function $F_U(\xx)$ is not separable.}
    \label{fig:instances}
\end{figure}

Although the function $F(\xx)$ has exponentially many local minima, it is separable in each independent variable $x_k$. The separability structure of $F(\xx)$ can be easily detected from its closed-form expression. Even if the closed-form expression is not available and the function is given via a black-box query model, a simple classical algorithm such as coordinate descent can split the problem into $d$ independent sub-problems, each can be solved in constant time.

To create hard optimization instances for classical algorithms, we apply a random rotation to \textit{hide} the separability structure in $F(\xx)$. More precisely, let $U$ be an orthogonal matrix and we define a new function,
\begin{align}\label{eqn:FU}
    F_U(\xx) \coloneqq F(U\xx).
\end{align}
The newly generated function $F_U(\xx)$ still has $2^d$ local minima as the rotation preserves the geometry of the optimization landscape. However, unlike the original function $F(\xx)$, the rotated function $F_U(\xx)$ is no longer separable (see \fig{instances}C, D).

More importantly, the random rotations on nonconvex functions are difficult to revert. For example, if the double-well function $w(x)$ is a degree-4 polynomial, so is the rotated function $F_U(\xx)$ because the rotation by $U$ is an affine transformation. However, to learn the rotation $U$ from the closed-form expression of $F_U(\xx)$ amounts to diagonalizing a 4-tensor. Diagonalizing general 4-tensors is known to be \texttt{NP}-hard~\cite{hillar2013most}. Therefore, it does not appear to be a trivial task for classical algorithms to revert the rotation and decompose $F_U(\xx)$ into $d$ separate sub-problems.

\begin{definition}[Nonconvex optimization instances]\label{defn:instances}
    Given an appropriate choice of the 1D double-well function $w(x)$ (more details available in \sec{instances}), we define a class of optimization instances $\mathscr{C}(w, d)$ for dimension $d \ge 1$:
\begin{align}
    \mathscr{C}(w, d) \coloneqq \{F_U(\vect{x}): U \in \mathcal{Q}(d)\},
\end{align}
where $F_U(\vect{x})$ is defined in \eqn{FU} and $\mathcal{Q}(d)$ is the set of all $d$-by-$d$ orthogonal matrix:
$$\mathcal{Q}(d) \coloneqq \{U \in \R^{d\times d}: U U^\top = I\}.$$
\end{definition}

\subsection{Main theoretical result}
Due to our construction of the optimization instances, any instance $f \in \mathscr{C}(w,d)$ has $2^d$ local minima but a unique global minimum $f(\xx^*)$. We say $\xx \in \R^d$ is an $\delta$-\textit{approximate} solution to the optimization problem $f$ if $\|\xx - \xx^*\| \le \delta$. Our main theoretical contribution is to prove that QHD can find a $\delta$-approximate solution to any $f \in \mathscr{C}(w,d)$ in polynomial time.

We assume our quantum algorithm only has access to the \textit{quantum evaluation oracle} (i.e., zeroth-order oracle), which is defined as a unitary map $O_f$ on $\R^d\otimes \R$ such that for any $\ket{\xx} \in \R^d$,
\begin{align}
    O_f\left(\ket{\xx}\otimes\ket{0}\right) = \ket{\xx}\otimes\ket{f(\xx)}
\end{align}
Note that the quantum evaluation oracle can be coherently accessed. Namely, for any $m \in \N$, let $\ket{\xx_1}, \dots,\ket{\xx_d} \in \R^d$ and $\mathbf{c} \in \C^m$ such that $\|\mathbf{c}\| = 1$, we have
\begin{align}
    O_f\left(\sum^m_{j=1}c_j\ket{\xx_j}\otimes\ket{0}\right) = \sum^m_{j=1}c_j\ket{\xx_j}\otimes\ket{f(\xx_j)}.
\end{align}

Our main theoretical result is summarized in the following theorem, where the $\tbigO$ notation suppresses poly-logarithmic factors in $d$ and $1/\delta$.

\begin{theorem}[Informal version of \thm{main}] \label{thm:informal_main}
    Let $f\colon \R^d\to \R$ be an instance in the class $\mathscr{C}(w,d)$. For any $\delta > 0$, Quantum Hamiltonian Descent can produce an $\delta$-approximate solution to $f$ with probability at least $2/3$ using $\tbigO\left(\frac{d^{3}}{\delta^2}\right)$ quantum queries to $f$ and additional $\tbigO\left(\frac{d^{4}}{\delta^2}\right)$ 1- and 2-qubit elementary gates.
\end{theorem}

\subsection{Empirical study of classical optimization algorithms}
In our empirical study, we select 6 representatives
from a comprehensive selection of state-of-the-art optimization algorithms in the following four major categories:

\begin{enumerate}
    \item Annealing-based and Monte-Carlo algorithms: dual annealing~\cite{xiang1997generalized}, basin-hopping~\cite{wales1997global};
    \item Gradient method: stochastic (or perturbed) gradient descent;
    \item Newton and quasi-Newton methods: interior-point method (Ipopt~\cite{wachter2006implementation}), sequential quadratic programming~\cite{nocedal1999numerical};
    \item Branch-and-bound algorithm: Gurobi~\cite{gurobi}.
\end{enumerate}

We test each optimization algorithm/solver on thousands of random instances from $\mathscr{C}(w,d)$ for a fixed degree-4 polynomial $w$ and different values of $d$, and then measure the run time and success rate.
Next, we compute the average \textit{time-to-solution} (TTS) for each classical algorithm/solver as a function of the dimension $d$. All the empirical scaling results are computed using the same software and hardware environment to ensure comparability of the collected run time data.

Our data suggests that, for all tested classical algorithms/solvers, the measured TTS scales as a super-polynomial function in dimension $d$. 
On the contrary, our theoretical result implies that the TTS of QHD scales polynomially in dimension $d$. More discussions and the visualization of empirical data are available in \sec{empirical}.

\subsection{Discussions on the implication and limitation}

The original QHD paper~\cite{leng2023quantum} has demonstrated a large-scale empirical implementation of QHD on a class of boxed-constrained quadratic programming problems up to 75 dimensions by using D-Wave machines as a quantum Ising Hamiltonian simulator and embedding the QHD's dynamics into the evolution of quantum Ising Hamiltonian with the so-called Hamiltonian embedding technique. 
However, due to the limitation of the D-Wave machine (e.g., decoherence and the limited connectivity that restricts the sparsity of the realizable box-constrained quadratic programming instances), QHD's empirical performance fails to beat state-of-the-art branch-and-bound solvers (such as Gurobi).  
In this paper, our newly constructed instances are empirically hard for Gurobi, which requires super-polynomial time to solve. 
However, our old technique fails to extend to implement QHD on the newly constructed instances on any existing quantum hardware. 
Instead, we theoretically prove that such instances can be solved in polynomial time on an ideal quantum computer. 

Our empirical study covers a wide spectrum of classical state-of-the-art optimization algorithms, where the numerical results suggest the run time of all tested classical algorithms scales super-polynomially in dimension $d$. 
This empirical evidence implies that QHD could potentially replace the role of these classical solvers and become one competitive \emph{off-the-shelf} optimization solution in practice. 
However, it does not exclude the possibility of efficient classical algorithms specially designed for our tested problem instances. 
We expect any such efficient classical algorithm, if exists, would, nevertheless, require a special design to leverage the hidden separable structure in the instances\footnote{ 
Note that \thm{informal_main} can be established with more general instances than those used in the empirical test. Any specially designed classical algorithm should aim to work with the general condition and its hidden structure.}.  Moreover, our empirical study is by no means complete as it is bound by limited computing resources and time.  In light of that, we make all of our empirical study publicly available and invite the whole community to test it further. 

We note further that our construction in the empirical test is unlikely to yield any super-polynomial oracle separation directly. 
This is because our optimization instances are $d$-dimensional degree-4 polynomials whose precise form can be fully determined with a polynomial, precisely $O(d^4)$, number of queries of the function value. 
Thus, any bigger than $\Omega(d^4)$ classical lower bound needs to be established in the plain model, which is likely beyond the reach of known techniques. 
We want to emphasize that our instances, however, could be more inspiring about concrete features of optimization problems that can be efficiently solved by QHD, compared with oracle separations, if possible, obtained via some technical routes. 
For example, any reduction from instances efficiently solvable by Shor's algorithm to optimization ones will likely not reveal any new quantum capability, as it is essentially running Shor's algorithm with an optimization disguise.

Finally, we want to emphasize that our specific construction of the hard optimization instances, however, should not be treated as a limitation of the applicability of the QHD algorithm, but rather a limitation of our method in analyzing QHD's behavior on general optimization instances. 
Indeed, as empirically evaluated in~\cite{leng2023quantum}, QHD could perform well on optimization instances much beyond the specific construction used in this paper. 
Analyzing the QHD's performance on these instances would require investigating the spectrum of Schr\"{o}dinger's operator with general potential energy functions, which we believe is an interesting but challenging problem in pure math. 

\paragraph{Code Availability.} 
The source codes of the classical empirical study are available at \url{https://github.com/lwins-lights/QHD/tree/main/classical}.

\paragraph{Notation.}
We let $\H = L^2(\R^d)$ and $\D = C^\infty_0(\R^d)$. The Hilbert space $\H$ is equipped with the standard inner product $\langle u, v \rangle = \int_{\R^d} \Bar{u}v~\d x$, which induces the $L^2$-norm in $\H$: $\|u\|^2 = \langle u, u \rangle$.

\subsection*{Acknowledgment}
We thank Scott Aaronson, Fernando Brandao, Shouvanik Chakrabarti, Alexander Dalzell,  Lei Fan, Aram Harrow,  Tongyang Li, Jin-Peng Liu, and Daochen Wang for insightful feedbacks. We also thank Lei Fan for helpful discussions on the empirical study with Gurobi. 
This work was partially funded by the U.S. Department of Energy, Office of Science, Office of Advanced Scientific Computing Research, Accelerated Research in Quantum Computing under Award Number DE-SC0020273, the Air Force Office of Scientific Research under Grant No. FA95502110051, the U.S. National Science Foundation grant CCF-1816695 and CCF-1942837 (CAREER), and a Sloan research fellowship.

\section{Preliminaries}\label{sec:prelim}
\subsection{Convexity and spectral gap}
Our theoretical analysis in this work heavily relies on the spectrum of the Schr\"odinger operator,
\begin{align}\label{eqn:operator}
    \hat{H} = -\nabla^2 + f,
\end{align}
where $\nabla^2$ is the Laplacian operator in the real space $\R^d$, and $f\colon \R^d \to \R$ is a continuous objective function that we want to minimize. The physical meaning of $f$ is a ``potential field'' that confines the motion of a quantum particle.

When $f$ is non-negative and it diverges at infinity (i.e., $\lim_{\|x\|\to \infty} f(x) = \infty$), the standard theory of elliptic operators~\cite[Theorem~10.7]{hislop2012introduction} guarantees that the spectrum of $\hat{H}$ is discrete and all eigenvalues are positive. We can arrange the eigenvalues of $\hat{H}$ in ascending order: $E_0 < E_1 \le \dots$. The difference between the first two eigenvalues (i.e., $E_1 - E_0$) is often called the \textit{spectral gap}. It is well-known that the spectral gap of the Schr\"odinger operator~\eqn{operator} has a deep connection to the convexity of the potential field $f$. Since the eighties, several authors (van den Berg~\cite{van1983condensation}, Ashbaugh and Benguria~\cite{ashbaugh1989optimal}, and Yau~\cite{yau1986nonlinear}) have independently observed that the spectral gap of a Schr\"odinger operator with a convex potential field has a lower bound $3\pi^2/D^2$, where $D$ is the diameter of the domain on which the Schr\"odinger operator is defined. This observation, known as the Fundamental Gap Conjecture, has been proven by Andrews and Clutterbuck~\cite{andrews2011proof} in 2011 using an involved analysis of the heat equation.

We now discuss a simple example known as the quantum harmonic oscillator. A quantum harmonic oscillator is described by a Schr\"odinger operator with a quadratic potential field,
\begin{align}
    H = \hbar^2 \left(-\frac{1}{2}\nabla^2\right) + \left(\frac{1}{2}\omega^2 x^2\right),
\end{align}
where $\omega$ is a fixed frequency and $\hbar$ is the so-called Planck constant. The eigenvalues of the quantum harmonic oscillator are given by $E_n = (2n+1)\frac{\hbar}{2} \omega$. We refer the readers to \cite{griffiths2018introduction} for more discussions on quantum harmonic oscillators.

If we consider the rescaled operator,
\begin{align*}
    \frac{1}{\hbar}H = \hbar \left(-\frac{1}{2}\nabla^2\right) + \frac{1}{\hbar}\left(\frac{1}{2}\omega^2 x^2\right),
\end{align*}
we find the eigenvalues of $H/\hbar$ are $E'_n = (n+1/2)\omega$, which only depend on the curvature of the potential field regardless of the value of $\hbar$. The spectral gap of the rescaled quantum harmonic oscillator $H/\hbar$ seems to be an invariant quantity determined by the geometry of the quadratic potential field.

Inspired by the case of quantum harmonic oscillators, we define a one-parameter family of Schr\"odinger operators for a given objective function $f$,
\begin{align}\label{eqn:sch_ops}
    \hat{H}(\lambda) = \frac{1}{\lambda} \left(-\frac{1}{2}\nabla^2\right) + \lambda f,
\end{align}
where $\lambda > 0$ is an interpolating parameter and we assume $f$ is non-negative and diverges at infinity. For any positive $\lambda$, the spectrum of $\hat{H}(\lambda)$ is discrete and we arrange the eigenvalues of $\hat{H}(\lambda)$ in ascending order,
\begin{align*}
    E_0(\lambda) < E_1(\lambda) \le E_2(\lambda) \le \dots.
\end{align*}
The spectral gap of $\hat{H}(\lambda)$ is denoted by 
\begin{align}
    \Delta(\lambda) \coloneqq E_1(\lambda) - E_0(\lambda).
\end{align}

If $f$ is a convex quadratic function, $\hat{H}(\lambda)$ is precisely a quantum harmonic oscillator. According to our discussion above, all the eigenvalues of the operator $\hat{H}(\lambda)$ only depend on $f$, regardless of the value of $\lambda$. Nevertheless, when $f$ is nonconvex and non-quadratic, the eigenvalues of $\hat{H}(\lambda)$ change along with $\lambda$. The spectral gap $\Delta(\lambda)$ depends on the nonconvexity of $f$~\cite{yau2009gap,andrews2011proof} and it also changes with $\lambda$. In particular, when $f$ is a double-well function (described by a quartic polynomial), the spectral gap can be made arbitrarily small by tuning the coefficients in the quartic polynomial~\cite{harrell1980double}. 

Although the spectral gap $\Delta(\lambda)$ could change drastically with $\lambda$ when $f$ is nonconvex, we observe that $\min_{\lambda > 0}\Delta(\lambda)$ is usually not arbitrarily small. In \sec{gap}, we study the spectral gap of the Schr\"odinger operator $\hat{H}(\lambda)$ with a (nonconvex) double well potential function using numerical methods. It turns out that the spectral gap $\Delta(\lambda)$ of nonconvex $f$ is robust under interpolation and its minimal value (i.e., $\min_{\lambda>0}\Delta(\lambda)$) appears to be determined solely by the geometry of $f$. Based on this observation, we prove that QHD can solve 1D nonconvex problems with two local minima (see \lem{1_dim_adiabatic}). This result is later generalized to high dimensions in \prop{d_dim_adiabatic} and eventually leads to our main theoretical result (see \thm{main}).

\subsection{Ground states of the Schr\"odinger operator}
In the Schr\"odinger operator \eqn{operator}, there are two major components: the kinetic operator $-\nabla^2$ and the potential operator $f$. When the kinetic operator dominates, the low-energy subspace of the Schr\"odinger operator tends to behave like that of a free particle. In another case, if the potential $f$ plays a leading role, the ground state of the Schr\"odinger operator exhibits properties as like a Gaussian distribution. For example, if we consider the following 1D Schr\"odinger operator,
\begin{align}\label{eqn:interpolated_operator}
    \hat{H}(\lambda) = \frac{1}{\lambda}\left(-\frac{1}{2}\nabla^2 \right) + \lambda f,
\end{align}
where $f\colon \R \to \R$ is a smooth one-dimensional potential function with a single, non-degenerate zero at $x^* = 0$: $f(0) = 0$, $\nabla f(0) = 0$, and $f''(0) > 0$. The function $f$ admits a Taylor expansion near the global minimum: 
$$f(x) = \frac{1}{2}f''(x^*) (x-x^*)^2 + O(x^3).$$
 As $\lambda$ becomes large, the potential landscape observed by the quantum particle is essentially the quadratic part of the potential $f$ near $x^*$. Let $\gamma = \sqrt{f''(x^*)}$, we can define a comparison Schr\"odinger operator
 \begin{align}
    \hat{K}(\lambda) = \frac{1}{\lambda}\left(-\frac{1}{2}\nabla^2 \right) + \lambda \left(\frac{1}{2}\gamma^2 (x-x^*)^2 \right).
\end{align}
This new operator $\hat{K}(\lambda)$ is a quantum harmonic oscillator and it gives a nice approximation of the low-energy subspace of $\hat{H}(\lambda)$ for large $\lambda$.

\begin{theorem}[Theorem 4.1, \cite{simon1983semiclassical}]\label{thm:semiclassical}
    Assume $f(x)$ is polynomially bounded (i.e., $|f(x)| \le C(1+|x|^m)$ for some constant $C$ and a fixed integer $m$). Let $n$ be a fixed integer such that $E_n(\lambda)$ is a simple\footnote{A \emph{simple} eigenvalue means the multiplicity is $1$. In other words, the corresponding eigenspace is non-degenerate.} eigenvalue of $\hat{H}(\lambda)$. Let $\Phi_n(\lambda)$ be the corresponding eigenfunction. Denote $e_n$ and $\phi_n(\lambda)$ to be the $n$-th eigenvalue and eigenfunction of $\hat{K}(\lambda)$. Then, for sufficiently large $\lambda$, we have that
    \begin{align}
        E_n(\lambda) - e_n = \bigO(\lambda^{-1}),~
        \|\Phi_n(\lambda) - \phi_n(\lambda)\| = \bigO(\lambda^{-1/2}).
    \end{align}
\end{theorem}

The ground state of the quantum harmonic oscillator $\hat{K}(\lambda)$ is a Gaussian state
\begin{align}\label{eqn:gaussian-ground-state}
    \phi_0(\lambda) = \left(\frac{\lambda \gamma}{\pi}\right)^{1/4} e^{-\frac{\gamma \lambda (x-x^*)^2}{2}},
\end{align}
and the eigenvalues of $\hat{K}$ are given by $e_n = \left(n + \frac{1}{2}\right)\gamma$.

As $\lambda$ increases, the Gaussian state $\phi_0(\lambda)$ is localized in the vicinity of the global minimizer $x^*$ and its tail (i.e., the probability of landing far away from $x^*$) is exponentially small. One would naturally guess that $\Phi_0(\lambda)$, the actual ground state of the Schr\"odinger operator \eqn{interpolated_operator}, also has a very small tail for large $\lambda$. However, this is not immediately clear from \thm{semiclassical}. In \sec{complexity}, we will prove a stronger result to characterize the smallness of the tail of $\Phi_0(\lambda)$ (see \lem{subgaussian}). It turns out that, given a fast-growing potential field $f$, the ground state $\Phi_0(\lambda)$ corresponds to a \textit{sub-Gaussian} distribution whose variance scales proportional to $1/\lambda$.

\subsection{Quantum adiabatic theorem for unbounded Hamiltonian}
Given a quantum Hamiltonian $H(t)$, $t \in [0, t_f]$, we consider the dynamics described by the Schr\"odinger equation,
\begin{align}\label{eqn:adiabatic}
    i \epsilon \frac{\d}{\d t}\ket{\psi^\epsilon(t)} = H(t)\ket{\psi^\epsilon(t)},
\end{align}
subject to an initial state $\ket{\psi^\epsilon(0)} = \ket{\psi_0}$. We suppose that $U^\epsilon(t)$ is the propagator of the dynamics described by \eqn{adiabatic}, so the solution at time $t$ is given by
\begin{align}
    \ket{\psi^\epsilon(t)} = U^\epsilon(t) \ket{\psi_0}.
\end{align}

The parameter $\epsilon > 0$ controls the time scale on which the quantum Hamiltonian $H(t)$ varies. For a small $\epsilon$, the system evolves slowly and the dynamics are relatively simple: if the system begins with an eigenstate of $H(0)$, it remains close to an eigenstate of $H(t)$. This process is called \textit{adiabatic quantum evolution}.

Formally, We define a new Hamiltonian operator,
\begin{align}\label{eqn:adb_ham}
    H_a(t) = H(t) + i \epsilon[\dot{P},P],
\end{align}
where $P(t)$ is a rank-1 projector onto the ground state of $H(t)$, $\dot{P}$ is the time derivative of $P$, and $[\cdot,\cdot]$ is the commutator. Let $U^\epsilon_a(t)$ be the propagator of the quantum dynamics generated by \eqn{adb_ham}, i.e., $U^\epsilon_a(t)$ is given as the solution of
\begin{align}
    i \epsilon \frac{\d}{\d t} U^\epsilon_a(t) = H_a(t) U^\epsilon_a(t),
\end{align}
subject to $U^\epsilon_a(0) = \mathbf{1}$. The propagator $U^\epsilon_a(t)$ is called the \textit{adiabatic intertwiner} because it preserves the ground-energy subspace of $H(t)$~\cite{kato1950adiabatic}. Precisely, we have that 
\begin{align}
    U^\epsilon_a(t) P(0) = P(t) U^\epsilon_a(t).
\end{align}

The quantum adiabatic theorem states that adiabatic intertwiner (i.e., the exact adiabatic propagator) $U^\epsilon_a(t)$ is a good approximation of the actual propagator $U^\epsilon(t)$ for sufficiently small $\epsilon$. Various formulations of the quantum adiabatic theorem exist in the literature, while most of them assume the Hamiltonian is a bounded linear operator. In Quantum Hamiltonian Descent, however, the Hamiltonian is an unbounded operator defined in the real space $\R^d$. Here, we introduce a quantum adiabatic theorem for unbounded Hamiltonian~\cite{teufel2003adiabatic,mozgunov2023quantum}, which allows us to have a neat analysis without discretizing the unbounded operator in QHD.

For a self-adjoint operator $H(t)$ with discrete spectrum at any $t$, the spectral gap of $H(t)$, denoted by $\Delta(t)$, is the difference between the first two eigenvalues of $H(t)$ at time $t$.

\begin{theorem}[Theorem 2.1~\cite{mozgunov2023quantum}]\label{thm:adiabatic_unbounded}
    Assume that for all $t \in [0, t_f]$, there exist positive numbers $c_0$ and $c_1$ such that 
    \begin{align}\label{eqn:bounded_Hdot}
        \dot{H}^2 \le c_0 + c_1 H^2.
    \end{align}
    Moreover, we assume $H > 0$ and the spectral gap $\Delta(t)$ has a uniform lower bound, i.e., $\Delta(t) \ge \Delta_{\min} > 0$.
    We denote $\phi_0$ as the ground state of $H(0)$.
    Then, we have that
    \begin{align}\label{eqn:adb_err}
        \|\left(U^\epsilon(t_f) - U^\epsilon_a(t_f)\right)\phi_0\| \le \epsilon \theta t_f,
    \end{align}
    where $\theta$ is a constant that only depends on $c_0$, $c_1$ and $\Delta_{\min}$.
\end{theorem}
\begin{remark}
    In \thm{adiabatic_unbounded}, we rescale the time variable compared to the original version in \cite[Theorem 2.1]{mozgunov2023quantum}. After the time rescaling, the adiabatic error from the original theorem reads $b = \frac{\theta t_f}{1/\epsilon} = \epsilon \theta t_f$, which is precisely what we have in \eqn{adb_err}. 
\end{remark}

\subsection{Quantum simulation of Schr\"odinger equations}\label{sec:q_simulation}
In Quantum Hamiltonian Descent~\cite{leng2023quantum}, a major subroutine is to simulate the Schr\"odinger equation, a task known as \textit{quantum simulation}.
Quantum simulation is a prominent application of quantum computation, and several efficient quantum algorithms for simulating real-space quantum dynamics have been proposed, including \cite{wiesner1996simulations,zalka1998efficient,an2021time,childs2022quantum,an2022time}. In this paper, we employ a quantum simulation algorithm due to Childs, Leng, Li, Liu, and Zhang~\cite{childs2022quantum}. The complexity of this algorithm has near-optimal dependence in the dimension $d$ and accuracy $\eta$ by leveraging the pseudo-spectral representation and interaction-picture quantum simulation.
We consider the Schr\"odinger equation over the time interval $[t_0, t_1]$ for a given time-dependent potential $V(x,t)$,
\begin{align}\label{eqn:schrodinger}
    i \frac{\partial}{\partial t}\ket{\Psi(x, t)} = \left[-\frac{1}{2}\nabla^2 + V(x, t)\right]\ket{\Psi(x, t)},
\end{align}
where we specify $\Omega = [-M, M]^d$ for a sufficiently large $M$ and $V(x, t)\colon \Omega \times [t_0, t_1]\to \R$ is a time-dependent potential function, $\Psi(x,t)\colon \Omega \times [t_0, t_1] \to \C$ is the wave function subject to certain initial condition and the periodic boundary condition.

In \cite[Theorem 8]{childs2022quantum}, the complexity of simulating \eqn{schrodinger} involves an additional parameter $g'$ that depends on the regularity of the wave function $\Psi(x,t)$. We slightly improve this result by assuming the initial condition is \textit{analytic}, which is the case in QHD.

\begin{theorem}\label{thm:spectral-method}
    Suppose the potential field $V(x, t)$ is bounded, smooth in $x$ and $t$, and periodic in $x$. Moreover, we assume the initial data $\Psi_0(x)$ is analytic on $\Omega$ and $V$ is $L$-Lipschitz in $t$. We define $\|V\|_{\infty,1}\coloneqq \int^{t_1}_{t_0}\|V(\cdot,t)\|_{\infty}~\d t$.
    Then, the Schr\"odinger equation \eqn{schrodinger} can be simulated for time $t\in [t_0, t_1]$ up to accuracy $\eta$ with the following cost:
    \begin{enumerate}
        \item Queries to the quantum evaluation oracle $O_V$: $\bigO\left(\|V\|_{\infty,1} \frac{\log(\|V\|_{\infty,1}/\eta)}{\log\log(\|V\|_{\infty,1}/\eta)}\right)$,
        \item 1- and 2-qubit gates: 
        $$\bigO\left(\|V\|_{\infty,1}\left(\poly(z) + \log^{2.5}\left(L\|V\|_{\infty,1}/\eta\right) + d\log\log(1/\eta)\right)\frac{\log(\|V\|_{\infty,1}/\eta)}{\log\log(\|V\|_{\infty,1}/\eta)}\right).$$
    \end{enumerate}
\end{theorem}

The proof of the theorem is available in \append{simulation}. Note that the quantum simulation algorithm in \cite{childs2022quantum} also requires two quantum oracles other than $O_V$, namely, the \textit{inverse change-of-variable} oracle $O_{\mathrm{inv}}$ and the \textit{max-norm} oracle $O_{\mathrm{norm}}$ (see Lemma 5~\cite{childs2022quantum}). In QHD, the potential function $V(x,t)\coloneqq \varphi(t) V(x)$, where $\varphi(t)$ is a time-dependent function described by a closed-form formula (see~\eqn{varphi}) and $V(x)$ is given by $O_V$. In this case, the two oracles $O_{\mathrm{inv}}$ and $O_{\mathrm{norm}}$ are efficiently implemented without querying the function $V$.

\section{Construction of nonconvex optimization instances}\label{sec:instances}
\subsection{The spectral gap of asymmetric double well}\label{sec:gap}
In this section, we use a numerical example to illustrate the theory of Schr\"odinger operators that we discussed in \sec{prelim}. We consider the following 1-parameter family of Hamiltonian operators
\begin{align}\label{eqn:lambda_ham}
    \hat{H}(\lambda) = \frac{1}{\lambda}\left(-\frac{1}{2}\nabla^2\right) + \lambda w(x),
\end{align}
where $w(x)\colon \R \to \R$ is a one-dimensional double-well function given by a degree-4 polynomial:
\begin{align}\label{eqn:wx}
    w(x) = x^4 - \left(x - \frac{1}{32}\right)^2 - c,
\end{align}
where $c \approx -0.296$ such that the global minimum of $w(x)$ is zero.

\begin{figure}[h]
    \centering
    \includegraphics[width=14cm]{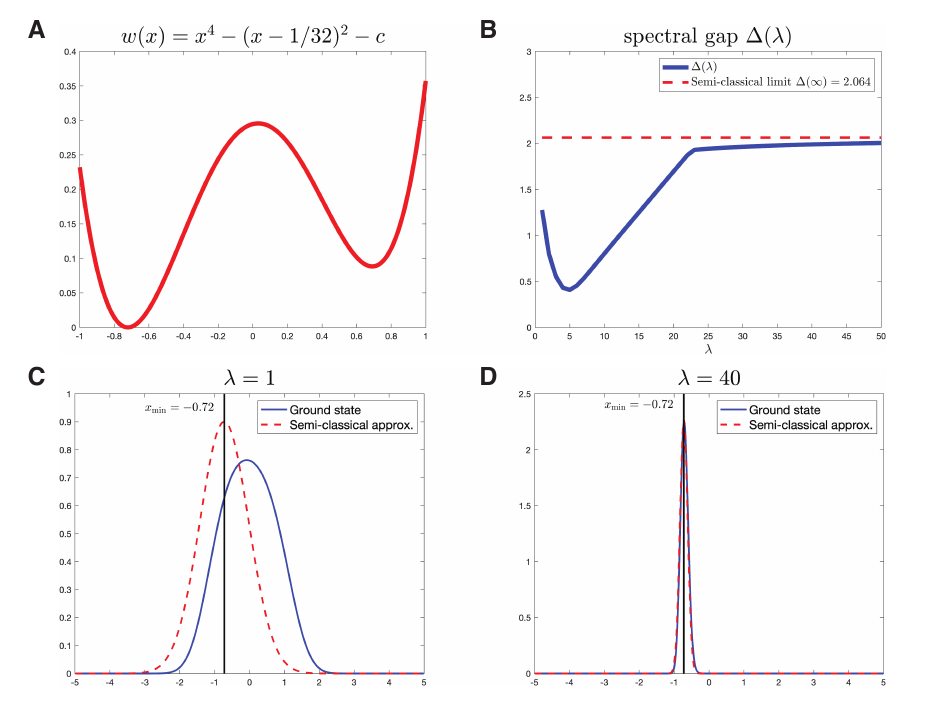}
    \caption{\small \textbf{The spectral gap and ground state of asymmetric double well.}
    \textbf{A.} Plot of the double well function $w(x)$.
    \textbf{B.} The spectral gap $\Delta(\lambda)$ of the Hamiltonian operator $\hat{H}(\lambda)$ as a function of $\lambda$.
    \textbf{C, D.} The blue curve (or the red dashed curve, resp.) represents the ground state (or the harmonic approximation of the ground state, resp.) of the Hamiltonian operator $\hat{H}(\lambda)$ for $\lambda = 1, 40$. The $x$ axis represents the one-dimensional space and the $y$ axis shows the amplitude of the quantum states. When $\lambda$ becomes larger, the harmonic approximation becomes more accurate.
    }
    \label{fig:double_well}
\end{figure}

In \fig{double_well}A, we plot the graph of the function $w$ between $x = -1$ and $x = 1$. This function has two local minima: 
\begin{align*}
    x_1 &\approx -0.722,~w(x_1) = 0,\\
    x_2 &\approx 0.691,~w(x_2) = 0.088.
\end{align*}
For any $\lambda > 0$, the spectral theory implies that the spectrum of $\hat{H}(\lambda)$ is discrete. We denote the spectrum of the Hamiltonian $\hat{H}(\lambda)$ as $\{E_0(\lambda) < E_1(\lambda) \le \dots\}$. For large $\lambda$, the first few eigenvalues of $\hat{H}(\lambda)$ can be computed using the semi-classical approximation. Note that $w$ has a positive Hessian at the global minimum: $w''(x_1) \approx 4.5092$. Let $\gamma^2 = w''(x_1)$. For $\lambda \gg 1$, by \thm{semiclassical}, we have that
\begin{align}
    E_0(\lambda) = \frac{1}{2}\gamma + \bigO(\lambda^{-1}),\quad E_1(\lambda) = \frac{3}{2}\gamma + \bigO(\lambda^{-1}),
\end{align}
which implies that the spectral gap in the large $\lambda$ limit is
\begin{align}\label{eqn:w_limit}
    \lim_{\lambda \to \infty}\Delta(\lambda) = \gamma \approx 2.064.
\end{align}

We employ numerical methods to compute the spectral gap of the operator $\hat{H}(\lambda)$. Our numerical results confirm our estimate of the spectral gap in the limit of large $\lambda$. 
In \fig{double_well}B, we show the spectral gap $\Delta(\lambda) = E_1(\lambda) - E_0(\lambda)$ as a function of $\lambda$. For large $\lambda$, it is clear that the spectral gap eventually converges to the predicted value $\gamma$ (indicated by the red dashed line). We also find that the spectral gap $\Delta(\lambda)$ achieves its minimum $\Delta_{\min} \approx 0.408$ at around $\lambda = 5$. For $\lambda > 5$, the spectral gap $\Delta(\lambda)$ starts to increase and gradually approaches the semi-classical limit.
At first glance, it is a bit surprising to see that the spectral gap $\Delta(\lambda)$ has a lower bound for all $\lambda > 0$ even when the potential function $w$ is nonconvex. This means that the minimal spectral gap $\Delta_{\min}$ is actually an intrinsic property that only depends on the geometry of the double-well potential $w$ but not the interpolating parameter $\lambda$. Therefore, it is reasonable to formalize our observation by defining a class of \textit{gapped} potential function $w$.

\begin{definition}[Gapped 1D potential]\label{defn:gappedness}
    Given a 1-dimensional twice differentiable function $w\colon \R \to \R$ such that $w$ is bounded from below and diverges to infinity as $|x| \to \infty$. For any $\lambda > 0$, let $\Delta(\lambda)$ denote the spectral gap of the Hamiltonian operator~\eqn{lambda_ham}.
    We say this function $w$ is \textbf{gapped} if the spectral gap is uniformly bounded by a positive constant $\Delta_{\min} > 0$, i.e., 
    \begin{align}
        \forall \lambda > 0\colon \Delta(\lambda) \ge \Delta_{\min} > 0.
    \end{align}
\end{definition}

We believe this gap condition holds for many asymmetric double-well potential functions. We will use these gapped asymmetric double wells to construct optimization instances with exponentially many local minima, and we will show that there exist efficient quantum algorithms (such as QHD) that can find the global solution in polynomial time.

Another numerical observation we made is that the semi-classical approximation works for not only the spectral gap but also the ground state. In panels C and D in \fig{double_well}, we plot the ground state of the Hamiltonian $\hat{H}(\lambda)$ (in blue solid line) and the semi-classical approximation given by \eqn{gaussian-ground-state} (in red dashed line) for $\lambda = 1, 40$. Clearly, when $\lambda$ is large, the Gaussian state given by the semi-classical theory is a very accurate approximation of the ground state.

\subsection{1D model problem}
We now formulate our 1D model problem.

\begin{definition}\label{defn:1}
    We call a non-negative function $w\colon \R \to \R$ a \textbf{well-formed} asymmetric double well if $w$ satisfies the following conditions:
    \begin{enumerate}
        \item $w$ has at least 2 local minima and a unique non-degenerate global minimum at $x = x^*$, i.e., $w'(x^*) = 0$ and $w''(x^*) \ge \gamma^2 > 0$;
        \item There exist positive numbers $C > c > 0$ such that $c|x-x^*|^4 \le w(x) \le C(1+|x-x^*|^4)$ for all $x\in \R$;
        \item $w\in C^{\infty}$ and $w$ is $\rho$-smooth, i.e., $\|w''\| \le \rho$;\footnote{Precisely speaking, a degree-4 polynomial is not $\rho$-smooth, as its Hessian is a quadratic function that diverges at infinity. Thanks to the localization of wave function in the quantum evolution (because $w$ grows as $|x|^4$), we may restrict the evolution to a compact domain in quantum simulation. In this case, the double-well potential function is effectively $\rho$-smooth.}
        \item $w$ is a gapped function in the sense of \defn{gappedness}.
    \end{enumerate}
\end{definition}

In \sec{gap}, we use numerical methods to study the function $w(x) = x^4 - \left(x - \frac{1}{32}\right)^2 - c$, where $c \approx -0.296$ such that $\min_x w(x) = 0$. It is clear that this $w(x)$ satisfies the first two conditions in \defn{1}. If we restrict $w(x)$ to a compact subset of $\R$ (which is the case in our quantum simulation), condition 3 is satisfied. Meanwhile, our numerical results in \fig{double_well} show that the spectral gap $\Delta(\lambda)$ achieves its minimum near $\lambda = 5$ and then $\Delta(\lambda)$ converges to the semi-classical limit, aligned with the theoretical prediction given by \thm{semiclassical}. Therefore, with the numerical evidence, we believe the function $w(x)$ is a \emph{well-formed} asymmetric double well. This function will later be used in our empirical study (see \sec{empirical}).  

\vspace{4mm}
When the function $w$ is a well-formed asymmetric double well, we can apply the adiabatic theorem for unbounded Hamiltonian (i.e., \thm{adiabatic_unbounded}) to prepare the ground state of the Schr\"odinger operator 
$$\hat{H}(\lambda) = \frac{1}{\lambda}\left(-\frac{1}{2}\nabla^2\right) + \lambda w(x)$$
for very large $\lambda$. 
Meanwhile, condition 2 in \defn{1} ensures that the ground state of $\hat{H}(\lambda)$ is sub-Gaussian (for details, see \lem{subgaussian}). These nice properties will allow us to find an approximate solution to the optimization problem $\min_x w$.

Provided with a pre-fixed parameter $\lambda_f$ (which we will discuss soon), we define a time-dependent function
\begin{align}\label{eqn:lambda_fun}
    \lambda(t) = e^{2t - \log(\lambda_f)},
\end{align}
Clearly, for $0 \le t \le \log(\lambda_f)$, we have that $1/\lambda_f \le \lambda(t) \le \lambda_f$.

\begin{lemma}\label{lem:1_dim_adiabatic}
    Let $w\colon \R \to \R$ be a well-formed asymmetric double well.
    For $0 \le t \le \log(\lambda_f)$, let $H(t) \coloneqq \hat{H}(\lambda(t))$. 
    Denote $\ket{\phi(t)}$ as the ground state of $H(t)$ for $t \in [0, \log(\lambda_f)]$.
    Let $\ket{\psi^\epsilon(t)}$ be the solution to the Schr\"odinger equation~\eqn{adiabatic} governed by $H(t)$ with the initial condition $\ket{\psi^\epsilon(0)} = \ket{\phi(0)}$. Then, for any $0 \le t \le \log(\lambda_f)$, we have that 
    \begin{align}
        \|\psi^\epsilon(t) - \tilde{\phi}(t)\| \le \epsilon \theta t,
    \end{align}
    where $\theta$ is an absolute constant that only depends on $f$. Here, $\ket{\tilde{\phi}(t)}$ only differs from $\ket{\phi(t)}$ by a global phase, i.e., 
    \begin{align}
        \tilde{\phi}(t) = e^{i\alpha(t)}\phi(t),
    \end{align}
    where $\alpha(t)$ is a real-valued function.
\end{lemma}
\begin{proof}
First, we show that $H(t)$ satisfies all conditions specified in \thm{adiabatic_unbounded}. Since $w$ is non-negative, $\hat{H}(\lambda) \ge 0$ for all positive $\lambda$. The gappedness of $w$ guarantees a uniform lower bound on the spectral gap of $H(t)$. We claim that 
\begin{align}\label{eqn:estimate2}
    \dot{H}^2 \le 4\rho + 4H^2.
\end{align}

By the definition of $H(t)$, we have that
\begin{align}
    \dot{H} = \frac{\dot{\lambda}}{\lambda}\left(\frac{1}{2 \lambda} \frac{\d^2}{\d x^2} + \lambda w\right) = 2\left(\frac{1}{2\lambda} \frac{\d^2}{\d x^2} + \lambda w\right).
\end{align}
For any test function $\psi \in C^\infty_0(\R)$, we consider the inner product:
\begin{align*}
    \langle \psi, [4H^2 - \dot{H}^2]\psi\rangle &= 4\langle \psi, H^2\psi\rangle -  \langle \psi,(H')^2\psi\rangle\\
    &= -4 \left(\langle \psi, (w \psi)''\rangle + \langle \psi, w\psi''\rangle\right)\\
    &= -4\left(\langle \psi, w'' \psi\rangle + 2\langle \psi, w'\psi'\rangle + 2\langle \psi, w\psi''\rangle\right).
\end{align*}
Using integral by parts, we have that 
\begin{align}
    \langle \psi, w\psi''\rangle = \langle w\psi, \psi''\rangle = - \langle w' \psi, \psi'\rangle - \langle w \psi', \psi'\rangle = - \langle \psi, w'\psi'\rangle - \langle \psi', w\psi'\rangle,
\end{align}
which implies that
\begin{align}
    \langle \psi, [4 H^2 - \dot{H}^2]\psi\rangle &= -4\langle \psi, w'' \psi\rangle + 4 \langle \psi', w \psi'\rangle.
\end{align}
By \defn{1}, $w$ is non-negative and $\rho$-smooth, so we have $\langle \psi', w \psi'\rangle \ge 0$ and $w'' \le \rho$. Therefore, we have that
\begin{align}
    \langle \psi, [4 H^2 - \dot{H}^2]\psi\rangle \ge -4\rho,
\end{align}
which implies \eqn{estimate2}. With all conditions satisfied, we invoke \thm{adiabatic_unbounded} and obtain that
\begin{align}\label{eqn:propagator_error}
    \|\left(U^\epsilon(t) - U^\epsilon_a(t)\right)\phi(0)\| \le \epsilon \theta t,
\end{align}
where $\theta$ is a constant that only depends on $w$. Note that 
\begin{align}\label{eqn:state_eps}
    U^\epsilon(t)\ket{\phi(0)} = \ket{\psi^\epsilon(t)},
\end{align}
and $U^\epsilon_a(t) \ket{\phi(0)}$ is always in the ground-energy subspace of $H(t)$, we denote 
\begin{align}\label{eqn:state_a}
    \tilde{\phi}(t) = U^\epsilon_a(t) \ket{\phi(0)}
\end{align}
which differs from $\ket{\phi(t)}$ by a global phase because the ground-energy subspace of $H(t)$ is 1-dimensional. Substituting \eqn{state_eps} and \eqn{state_a} to \eqn{propagator_error}, we complete the proof.
\end{proof}

\subsection{Toward high-dimensional nonconvexity} \label{sec:nd}
Let $d \ge 1$ be an integer and $w(x)$ be a well-formed asymmetric double well. We define the $d$-dimensional objective function $F$:
\begin{align*}
    F(x_1,\dots,x_d) = \sum^d_{k=1}w(x_k).
\end{align*}
If the 1D model function $w(x)$ has two local minima (e.g., a double-well potential), the function $F(\vect{x})$ has $2^d$ local minima. Let $U$ be an arbitrary $d$-by-$d$ orthogonal matrix, we define a new function $F_U(\xx) \coloneqq F(U\xx)$. Now, we consider the time-dependent Hamiltonian:
\begin{align}\label{eqn:d_Hs}
    H(t) = \frac{1}{\lambda(t)} \left(-\frac{1}{2}\nabla^2\right) + \lambda(t) F_U(\vect{x}),
\end{align}
where the function $\lambda(t)$ is defined in \eqn{lambda_fun}.

\begin{proposition}\label{prop:d_dim_adiabatic}
    For $0 \le t \le \log(\lambda_f)$, let $H(t)$ be defined as in \eqn{d_Hs}. Denote $\ket{\Phi(t)}$ as the ground state of $H(t)$ for $t \in [0, \log(\lambda_f)]$. Let $\ket{\Psi^\epsilon(t)}$ be the solution to the Schr\"odinger equation~\eqn{adiabatic} governed by $H(s)$ with the initial condition $\ket{\Psi^\epsilon(0)} = \ket{\Phi(0)}$. Then, for any $0 \le t \le \log(\lambda_f)$, we have that
    \begin{align}\label{eqn:d_dim_adiabatic}
        \|\Psi^\epsilon(t) - \tilde{\Phi}(t)\| \le \epsilon \theta d t,
    \end{align}
    where $\theta$ is an absolute constant that only depends on $w(x)$, and $\ket{\tilde{\Phi}(t)}$ only differs from $\ket{\Phi(t)}$ by a global phase.
\end{proposition}
\begin{proof}
    We consider the following Hamiltonian operator with the separable objective function $F(\vect{x})$,
    \begin{align}
        H_0(t) = \frac{1}{\lambda(t)} \left(-\frac{1}{2}\nabla^2\right) + \lambda(t) F(\vect{x}).
    \end{align}
    For $t\in[0,\log(\lambda_f)]$, let $\ket{\Psi^\epsilon_0(t)}$ be the solution to the Schr\"odinger equation,
    \begin{align}
        i \epsilon \ket{\Psi^\epsilon_0(t)} = H_0(t) \ket{\Psi^\epsilon_0(t)}.
    \end{align}
    If the initial state $\Psi^\epsilon_0(\vect{x}, 0) = \Psi^\epsilon(U^\top\vect{x}, 0)$, since the Laplacian is invariant under rotation, it is clear that  
    \begin{align}
        \Psi^\epsilon_0(\vect{x}, t) = \Psi^\epsilon(U^\top\vect{x}, t)
    \end{align}
    for all $t \in [0,\log(\lambda_f)]$. Similarly, the ground state of $H_0(t)$ is given by 
    \begin{align}
        \Phi_0(U^\top\vect{x}, t) = \Phi(U^\top\vect{x}, t),
    \end{align}
    where $\Phi(\xx, t)$ is the ground state of $H(t)$. Therefore, to show \eqn{d_dim_adiabatic}, it suffices to prove that there is a state $\ket{\tilde{\Phi}_0(t)}$ that differs from $\ket{\Phi_0(t)}$ by a global phase such that
    \begin{align}\label{eqn:rotated_back}
        \|\Psi^\epsilon_0(t) - \tilde{\Phi}_0(t)\| \le \epsilon \theta d t.
    \end{align}

    Next, we prove \eqn{rotated_back} using the triangle inequality. Note that the Hamiltonian operator $H_0(t)$ is separable in the sense that
    \begin{align}
        H_0(t) = \sum^d_{k = 1}H^{(k)}_0(t),
    \end{align}
    where for $k = 1,\dots, d$, the operator $H^{(k)}_0(t)$ is given by
    \begin{align}
        H^{(k)}_0(t) \coloneqq \frac{1}{\lambda(t)} \left(-\frac{1}{2}\frac{\partial^2}{\partial x^2_k}\right) + \lambda(t) w(x_k).
    \end{align}
    If we denote the ground state of $H^{(k)}_0(t)$ as $\ket{\phi_k(t)}$, the separability implies that the ground state of $H_0(t)$ is given by the product of the 1D ground states $\Phi_0(t) = \prod^d_{k=1}\phi_k(t)$.
    Similarly, the solution $\ket{\Psi^\epsilon_0(t)}$ is a product state given by $\Psi^\epsilon_0(t) = \prod^d_{k=1}\psi^\epsilon_k(t)$, where each $\psi^\epsilon_k(t)$ solves the 1D Schr\"odinger equation $i \epsilon \frac{\d}{\d t}\ket{\psi^\epsilon_k(t)} = H^{(k)}_0(t) \ket{\psi^\epsilon_k(t)}$. 
    By \lem{1_dim_adiabatic}, we have that 
    \begin{align}
        \|\psi^\epsilon_k(t) - e^{i\alpha_k(t)}\phi_k(t)\| \le \epsilon \theta t,
    \end{align}
    where each $\alpha_k(t)$ is a real-valued global phase. We define 
    \begin{align}
        \tilde{\Phi}_0(t) \coloneqq e^{i\sum^d_{k=1}\alpha_k(t)} \prod \prod^d_{k=1}\psi^\epsilon_k(t) = e^{i\sum^d_{k=1}\alpha_k(t)} \Phi_0(t).
    \end{align}
    Then, by applying the triangle inequality $d$ times, we end up with
    \begin{align}
        \|\Psi^\epsilon_0(t) - \tilde{\Phi}_0(t)\| \le \sum^d_{k=1}\|\psi^\epsilon_k(t) - e^{i\alpha_k(t)}\phi_k(t)\| \le \epsilon \theta d t.
    \end{align}
\end{proof}

\section{Quantum algorithms and complexity analysis}\label{sec:complexity}
\subsection{Main results}
Suppose that we fix a one-dimensional model problem $w(x)$ that satisfies \defn{1} (e.g., an asymmetric double-well function) and an integer $d \ge 1$. With $w$ and $d$, we construct a class of optimization instances $\mathscr{C}(w,d)$ as in \defn{instances}. Given an instance $f \in \mathscr{C}(w,d)$, we consider a Quantum Hamiltonian Descent algorithm with parameters $\delta, \lambda_f, \eta$ and $\eta$.

\begin{algorithm}[htbp]
    \caption{Quantum Hamiltonian Descent}
    \label{algo:qhd}
    Specify positive parameters $\delta$, $\lambda_f$, $\epsilon$, and $\eta$. Let $t_f = \log(\lambda_f)$ and $\lambda(t) = e^{2t - t_f}$. For $0 \le t \le t_f$, we define the QHD Hamiltonian over $\Omega$, $$H(t) = \frac{1}{\lambda(t)}\left(-\frac{1}{2}\nabla^2\right) + \lambda(t) f.$$\\
    Prepare the ground state of $H(0)$, denoted by $\ket{\Phi(0)}$.\\
    Use $\ket{\Phi(0)}$ as the initial state and simulate the quantum dynamics governed by the Schr\"odinger equation, 
    \begin{align}\label{eqn:qhd_params}
        i \epsilon \frac{\d}{\d t}\ket{\Psi^\epsilon(t)} = H(t) \ket{\Psi^\epsilon(t)},
    \end{align}
    up to an additive error $\eta$, i.e., a final state $\ket{\Psi^\epsilon_{\mathrm{sim}}(t_f)}$ is obtained such that $\|\Psi^\epsilon_{\mathrm{sim}}(t_f) - \Psi^\epsilon(t_f)\| \le \eta$.\\
    Measure the final state $\ket{\Psi^\epsilon_{\mathrm{sim}}(t_f)}$ using the position quadrature $\hat{\vect{x}}$ and return a $d$-dimensional vector $\vect{x}$.
\end{algorithm}

Recall that a solution $\xx$ is called a $\delta$-approximate solution if $\|\vect{x} - \vect{x}^*\| \le \delta$, where $\xx^*$ is the (unique) global minimizer of the optimization instance $f$.

\begin{theorem}\label{thm:main}
Given a well-formed asymmetric double well $w$ (see \defn{1}) and a fixed integer $d$, let $f \in \mathscr{C}(w,d)$ be an optimization instance. For a sufficiently small $\delta > 0$, we specify the choices of parameters:
\begin{align}
    \lambda_f \coloneqq \frac{\beta(2d + 6\log(3))}{\delta^2},~\epsilon \coloneqq \frac{1}{18 d \theta \log(\lambda_f)},~\eta \coloneqq \frac{1}{18},
\end{align}
where $\beta$ and $\theta$ are constants that only depend on $w$. Then, \algo{qhd} produces a $\delta$-approximate global solution to $f$ using 
$$\bigO\left(\frac{d^3}{\delta^2}\cdot\frac{\log^2(\frac{d^2}{\delta^2})}{\log\log(\frac{d^2}{\delta^2})}\right)$$
queries to $O_f$ and additional 
$$\bigO\left(\frac{d^3}{\delta^2}\left(d + \log^{2.5}\left(\frac{d}{\delta^2}\right)\right)\frac{\log^2(\frac{d^2}{\delta^2})}{\log\log(\frac{d^2}{\delta^2})}\right)$$
1- and 2-qubit gates.
\end{theorem}

\subsection{Initial state preparation}\label{sec:init_state_prep}
An important step in \algo{qhd} is to prepare the ground state of the initial Hamiltonian $H(0)$. Since the function $f$ grows as $\Theta(|x|^4)$, we can choose a large enough constant $M$ such that the quantum dynamics in $\R^d$ can be faithfully simulated over a compact domain $\Omega \coloneqq [-M, M]^d$ (with periodic boundary conditions). We denote $V(\xx)$ as the truncated objective function $f(\xx)$ on $\Omega$. Then, it is sufficient to prepare the ground state of the operator 
\begin{align*}
    H(0) = \lambda_f\left(-\frac{1}{2}\nabla^2\right) + \frac{1}{\lambda_f}V.
\end{align*}

Since we choose $\lambda_f \gg 1$, the potential operator $V/\lambda_f$ can be regarded as a perturbative term in the initial Hamiltonian $H(0)$. Meanwhile, the ground state of the kinetic operator $\left(-\frac{1}{2}\nabla^2\right)$ over the periodic domain $\Omega$ is a constant function that is easily prepared on a quantum computer as a uniform superposition state. Therefore, we may prepare the initial state of $H(0)$ from the ground state of the kinetic operator via a quantum adiabatic evolution.

More concretely, we can simulate a time-dependent Schr\"odinger equation with the following Hamiltonian,
$$\tilde{H}(t) = -\frac{1}{2}\nabla^2 + \frac{a(t)}{\lambda^2_f}V,$$
where $a(t)\colon [0,T]\to [0,1]$ is a monotonically increasing function such that $a(0) = 0$, $a(T) = 1$. Note that the ground state of $\tilde{H}(0)$ is a constant function and that of $\Tilde{H}(T)$ is the desired initial state in QHD. To ensure the quantum state stays in the ground-energy subspace, the evolution time $T$ usually depends on a polynomial of $1/\Delta$, where $\Delta$ is the minimal spectral gap of $\tilde{H}(t)$ for $0 \le t \le T$. For any $t \in [0, T]$, the potential term $\frac{a(t)}{\lambda^2_f}V$ is perturbative so we can derive a constant lower bound for the spectral gap using the standard perturbation theory~\cite{hislop2012introduction}. This means the adiabatic evolution time $T$ is independent of the problem dimension $d$. 

The query complexity of this quantum evolution is typically proportional to the simulation time $T$ and the $L^\infty$-norm of the perturbation $V/\lambda^2_f$. As $V$ is the sum of $d$ identity 1D model problems $w(x) \sim |x-x^*|^4$, we have $\|V\|_{\infty} \sim dM^4$. Together with $\lambda_f \sim d$, we have that $\|V\|_{\infty}/\lambda^2_f \sim M^4/d$, where $M$ is a large constant. It turns out that the initial state of $H(0)$ is prepared with $\bigO(1/d)$ queries to $O_f$ and an additional $\bigO(d)$ 1- and 2-qubit gates. Both the query complexity and the gate complexity are significantly lower than the upper bounds claimed in \thm{main}.

Also, we would like to remark that using the ground state of the kinetic operator $\left(-\frac{1}{2}\nabla^2\right)$ (i.e., a constant function) does not appear to harm the overall performance of QHD. In \sec{numerical_qhd}, we conduct an experiment of QHD using the uniform superposition state as the initial state. Numerical results suggest the same convergence rate as predicted by our theoretical argument.

\subsection{Proof of \thm{main}}
Our proof of \thm{main} relies on the following lemma on the ground state of the Schr\"odinger operator $\hat{H}(\lambda)$. This lemma gives an estimate on the tail probability of the ground state. In particular, we show that the ground state of $\hat{H}(\lambda)$ describes a quantum particle whose position is a sub-Gaussian random variable. The complete proof of \lem{subgaussian} can be found in \append{proof_subgaussian}.

\begin{lemma}[Sub-Gaussian ground state]\label{lem:subgaussian}
    Let $w$ be a well-formed asymmetric double well. Let $\phi_\lambda(x)$ be the ground state of the Hamiltonian 
    $$\hat{H}(\lambda) = \frac{1}{\lambda}\left(-\frac{1}{2}\nabla^2\right) + \lambda w(x).$$
    Then, for any $s \in \R$, we have
    \begin{align}
    \mathbb{E}_{X\sim|\phi_{\lambda}|^2}\left[e^{s(X-x^*)}\right] \le e^{s^2\sigma^2/2},
    \end{align}
    where $\sigma^2 = \beta/\lambda$, and $\beta$ is a constant that only depends on $w$.
\end{lemma}

Also, to leverage the quantum simulation algorithm in \thm{spectral-method}, we need to rescale the time in the Schr\"odinger equation \eqn{qhd_params}. In the following lemma, we give the rescaled quantum dynamics whose final state is precisely the same as in \eqn{qhd_params}. The proof of \lem{time_dilation} is given in \append{technical}.

\begin{lemma}[Time rescaling]\label{lem:time_dilation}
    The final state $\ket{\Psi^\epsilon(t_f)}$ of the quantum dynamics in \algo{qhd} can be obtained by simulating the following Schr\"odinger equation for $\frac{1}{\epsilon \lambda_f} \le s \le \frac{\lambda_f}{\epsilon}$,
    \begin{align}\label{eqn:effective_dynamics}
        i \frac{\d}{\d s}\ket{\Psi(s)} = \frac{1}{2}\left[-\frac{1}{2}\nabla^2 + \varphi(s) f\right]\ket{\Psi(s)},
    \end{align}
where
    \begin{align}\label{eqn:varphi}
        \varphi(s) \coloneqq \frac{1}{\left(-\epsilon s + \lambda_f + 1/\lambda_f\right)^2} \in [1/\lambda^2_f, \lambda^2_f].
    \end{align}    
\end{lemma}

\vspace{4mm}
Now, we are ready to prove \thm{main}. We define $B_\delta\coloneqq \{\xx\in\R^d:\|\xx - \xx^*\| \le \delta\}$, i.e., the $\delta$-ball entered at $\xx^*$. Let $1_{B_\delta}(\xx)$ be the indicator function of the set $B_{\delta}$, i.e., 
\begin{align}
    1_{B_\delta}(\xx) = \begin{cases}
    1 & (\|\xx - \xx^*\| \le \delta)\\
    0 & (\mathrm{otherwise}).
    \end{cases}
\end{align}
Then, for a quantum particle described by a quantum state $\ket{\Psi}$, its position is a random variable $X\in \R^d$ such that $X \sim |\Psi|^2$. The success probability can be evaluated through the expectation,
\begin{align}
    \Pr_{X\sim |\Psi|^2}[\|X - \xx^*\| \le \delta] = \bra{\Psi}1_{B_\delta}(\xx)\ket{\Psi} = \int_{\R^d} 1_{B_\delta}(x) |\Psi(\xx)|^2~\d \xx.
\end{align}

\begin{proof}
Let $\ket{\Phi(t_f)}$ be the ground state of the final Hamiltonian $H(t_f) = \frac{1}{\lambda_f}\left(-\frac{1}{2}\nabla^2\right) + \lambda_f f$. The modulus square of the quantum state, $|\Phi(t_f)|^2$, gives a probability density on $\R^d$. We denote $X \in \R^d$ as a random vector following the distribution $|\Phi(t_f)|^2$. Due to our construction of the optimization instances, it is clear that $Y \coloneqq U^\top(X - \xx^*)$ is a $d$-dimensional random vector with independent coordinates. Moreover, by \lem{subgaussian},
\begin{align}
    \mathbb{E}\left[e^{r^\top Y}\right] \le e^{\|r\|^2\sigma^2/2}
\end{align}
for all $r \in \R^d$, where $\sigma^2 = \beta/\lambda$ and $\beta$ is a constant that only depends on $w$. Then, by \thm{subgaussian_tail}, we have that
\begin{align}\label{eqn:concentration_estimate}
    \Pr\left[\|Y\|^2 > d \sigma^2 + 2\sqrt{cd\sigma^4} + 2 \sigma^2 c\right] \le e^{-c}.
\end{align}
We choose $c = 2\log(3)$ (so $e^{-c} = 1/9$) so that
\begin{align*}
    d \sigma^2 + 2\sqrt{cd\sigma^4} + 2 \sigma^2 c \le \frac{\beta(2d + 6\log(3))}{\lambda_f}.
\end{align*}
If specify the parameter $\lambda_f$ as in \thm{main}, it follows that
\begin{align}\label{eqn:res1}
     \Pr[\|X_1 - \xx^*\| > \delta] = \Pr[\|Y\| > \delta] \le \frac{1}{9}.
\end{align}

Let $\ket{\Psi^\epsilon(t_f)}$ be the final state of the dynamics~\eqn{qhd_params}. If we choose $\epsilon$ as in \thm{main}, by \prop{d_dim_adiabatic}, we have that
\begin{align}\label{eqn:res2}
    \|\Psi^\epsilon(t_f) - \Tilde{\Phi}(t_f)\| \le \frac{1}{18},
\end{align}
where $\tilde{\Phi}(t_f)$ only differs from the ground state $\ket{\Phi(t_f)}$ by a global phase so $|\Phi(t_f)|^2$ and $|\tilde{\Phi}(t_f)|^2$ give exactly the same probability distribution.

Now, if we simulate the quantum dynamics in \algo{qhd} up to an additive error $\eta = 1/18$, by the triangle inequality and \eqn{res2}, we have that
\begin{align}
    \|\Psi^\epsilon_{\mathrm{sim}}(t_f) - \Tilde{\Phi}(t_f)\| \le \frac{1}{9}.
\end{align}
Then, we deduce from \lem{robust_state} that
\begin{align}\label{eqn:res3}
    \left|\bra{\Psi^\epsilon_{\mathrm{sim}}(t_f)}1_{B_\delta}\ket{\Psi^\epsilon_{\mathrm{sim}}(t_f)} - \bra{\tilde{\Phi}(t_f)}1_{B_\delta}\ket{\tilde{\Phi}(t_f)}\right| \le \frac{2}{9}.
\end{align}
Since $\tilde{\Phi}(t_f)$ describes the same probability distribution as $\Phi(t_f)$, we combine \eqn{res1} and \eqn{res3} to yield the desired estimate,
\begin{align}
    \Pr_{X\sim |\Psi^\epsilon_{\mathrm{sim}}(t_f)|^2}\left[\|X - \xx^*\| > \delta \right] \le \frac{1}{3}.
\end{align}

It remains to figure out the overall complexity of the quantum simulation. Since the function $f$ grows as $\Theta(|x|^4)$, we can choose a large enough constant $M$ such that the quantum dynamics in $\R^d$ can be faithfully simulated over a compact domain $\Omega \coloneqq [-M, M]^d$ (with periodic boundary conditions). In what follows, we let $V(\xx)$ be the truncated objective function $f(\xx)$ on $\Omega$.

The initial state of the quantum simulation can be prepared using the procedures in \sec{init_state_prep}. The query/gate complexity of the initial state preparation is at most $\bigO(d)$.

By \lem{time_dilation}, simulating the quantum simulation in \algo{qhd} is equivalent to simulating the effective quantum dynamics~\eqn{effective_dynamics}. We use the quantum algorithm in \thm{spectral-method} to simulate~\eqn{effective_dynamics}. For $s \in J$, the time-dependent potential is $V(s,x) = \varphi(s) V(x)$, so we have
\begin{align}
    \|V\|_{\infty, 1} = \left(\int^{\lambda_f/\epsilon}_{1/\epsilon\lambda_f} \varphi(s)~\d x\right)\|V\|_{\infty} \le \frac{\lambda_f}{\epsilon}\|V\|_{\infty} = 27\theta\|V\|_{\infty}d \lambda_f \log(\lambda_f) \le \bigO\left(\frac{d^3}{\delta^2}\right).
\end{align}
Note that $\|V\|_{\infty} \le Cd M^4$ for an absolute constant $C> 0$. Similarly, we can prove that the time-dependent potential function $V(x,s)$ is Lipschitz in $s$ and the Lipschitz constant is
\begin{align}
    L = \max_{s\in J} \|\dot{V}\|_{\infty} \le 2 \frac{\lambda^3_f}{27d\theta \lambda(\lambda_f)}\|V\|_\infty \le \bigO\left(\frac{d^3}{\delta^6}\right).
\end{align}
Therefore, by \thm{spectral-method}, the quantum simulation task can be implemented with
\begin{align}
    \bigO\left(\frac{d^3}{\delta^2}\cdot\frac{\log^2(\frac{d^3}{\delta^2})}{\log\log(\frac{d^3}{\delta^2})}\right)
\end{align}
queries to the quantum evaluation oracle $O_f$ and additional
\begin{align}
    \bigO\left(\frac{d^3}{\delta^2}\left(d + \log^{2.5}\left(\frac{d^3}{\delta^4}\right)\right)\frac{\log^2(\frac{d^3}{\delta^2})}{\log\log(\frac{d^3}{\delta^2})}\right)
\end{align}
1- and 2-qubit gates.
\end{proof}

\subsection{Numerical simulation for 1D QHD}\label{sec:numerical_qhd}
As complementary to the theoretical result, we numerically simulate QHD for the 1D model problem $f(x) = x^4 - (x-1/32)^2 - c$ (where $c \approx -0.296$ so $f(x) \ge 0$). By \eqn{concentration_estimate}, the tail estimate of the final ground state $\ket{\Phi(t_f)}$ is exponentially small in the parameter $c$,
\begin{align}
    \Pr_{X\sim |\Phi(t_f)|^2}\left[\|X - x^*\| > \frac{\beta(2+3c)}{\lambda_f}\right] \le e^{-c}.
\end{align}
If we choose $c = \frac{1}{3}\left(\frac{\delta \lambda_f}{\beta} - 2\right)$ for a fixed $\delta > 0$, the failure probability can be expressed as a function of $\lambda_f$,
\begin{align}\label{eqn:failure}
     \Pr_{X\sim |\Phi(t_f)|^2}\left[\|X - x^*\| > \delta\right] \le \bigO\left(e^{-\frac{\delta}{3\beta} \lambda_f}\right).
\end{align}
Meanwhile, \thm{main} suggests that we need to maintain $\epsilon \sim 1/\log(\lambda_f)$ to make sure that the final state in the quantum simulation is $\epsilon$-close to the ground state $\ket{\Phi(t_f)}$. 

In our numerical experiment, we choose $\epsilon = \frac{1}{10\log(\lambda_f)}$ and simulate the Schr\"odinger equation \eqn{qhd_params} for $0 \le t \le t_f = \log(\lambda_f)$.
We use the uniform superposition state as the initial state $\ket{\Phi(0)}$. 
When $t$ reaches the final time $t = t_f$, we stop the quantum simulation and obtain the final state $\ket{\Psi^\epsilon_{\mathrm{sim}}(t_f)}$. For some $\lambda_f$, we compute the failure probability by
\begin{align*}
    \Pr_{X\sim |\Psi^\epsilon_{\mathrm{sim}}(t_f)|^2}[\|X-x^*\| > \delta] = 1 - \bra{\Psi^\epsilon_{\mathrm{sim}}(t_f)}1_{B_{\delta}}(x) \ket{\Psi^\epsilon_{\mathrm{sim}}(t_f)},
\end{align*}
where we fix $\delta = 0.1$. As shown in \eqn{failure}, the failure probability should decay exponentially as the parameter $\lambda_f$ increases.

\begin{figure}[htbp]
    \centering
    \includegraphics[width=8cm]{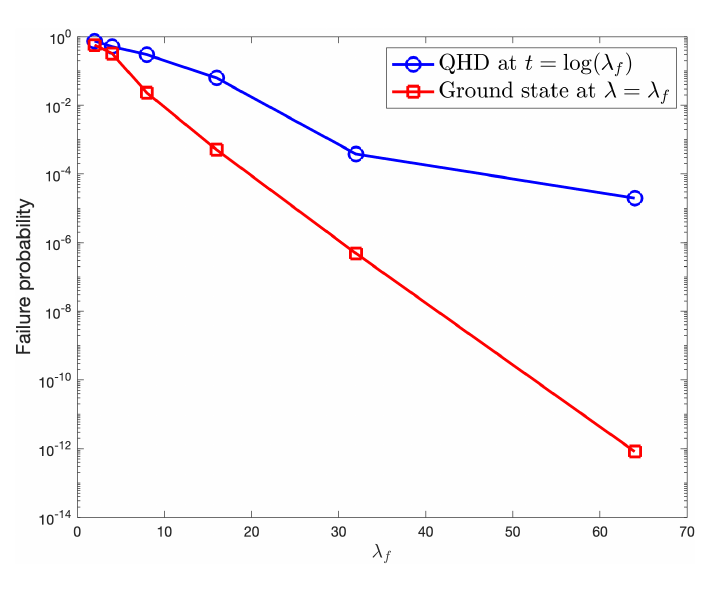}
    \caption{\small Numerical experiment for 1D QHD.}
    \label{fig:1d_qhd}
\end{figure}

In \fig{1d_qhd}, we show the failure probability of QHD as a function of $\lambda_f$, see the blue curve. As a reference, we also plot the failure probability suggested by the ground state of $H(t_f)$, see the red curve. 
As predicted by our theoretical argument, both curves decrease exponentially fast with respect to the parameter $\lambda_f$. 
Nevertheless, a significant gap between these two curves shows that the failure probability of QHD is higher than the exact final ground state. This might be a consequence of the adiabatic approximation error and the inexact choice of the initial state.

\section{Empirical study of classical algorithms}\label{sec:empirical}

In this section, we perform an empirical study to evaluate the performance of classical optimization algorithms for nonconvex optimization instances. We consider 6 state-of-the-art classical optimization algorithms (covering four major categories in the optimization literature):

\begin{enumerate}
    \item Annealing-based and Monte-Carlo algorithms:
    \begin{itemize}
        \item Dual annealing~\cite{xiang1997generalized}, which combines the generalization of CSA (Classical Simulated Annealing) and FSA (Fast Simulated Annealing);
        \item Basin-hopping~\cite{wales1997global}, which conceptually transforms the objective function into a collection of interpenetrating staircases (basins) by local optimization methods, and then explores these basins by a canonical Monte Carlo simulation that perturbs current coordinates and randomly accepts or rejects the new coordinates based on the corresponding basin value.
    \end{itemize}
    \item Gradient method: Canonical gradient descent with Gaussian noise added to the current coordinates in each step. This method is called perturbed gradient descent or stochastic gradient descent (SGD) in the literature. We stick to the latter hereafter.
    \item Newton and quasi-Newton methods:
    \begin{itemize}
        \item Interior-point method (Ipopt~\cite{wachter2006implementation});
        \item Sequential quadratic programming (SQP)~\cite{nocedal1999numerical}.
    \end{itemize}
    \item Branch-and-bound algorithm: Gurobi~\cite{gurobi}, which is a popular industrial-level optimization solver.
\end{enumerate}

\subsection{Methodology}\label{sec:methodology}

\paragraph{Problem Instances.}
The problem instances are introduced in ~\sec{nd}.
For a given positive integer $d$ and a $d$-by-$d$ orthogonal matrix $U$, we define a problem instance:
\begin{align*}
    \text{minimize} & \quad F_U(\vect{x}),
\end{align*}
where $\xx \in \R^d$, $F_U(\vect{x}) \coloneqq F(U\vect{x})$ and $F(\vect{x}) = F(x_1,\dots,x_d) \coloneqq \sum^d_{k=1}w(x_k)$. Here, we choose $w(x) = x^4 - \left( x - \frac{1}{32} \right)^2 - c$ with $c := \min_{x \in \R} \left\{ x^4 - \left( x - \frac{1}{32} \right)^2\right\} \approx -0.296 $.
The constant $c$ is introduced such that $\min_{\xx}F_U(\vect{x}) = 0$.

\paragraph{Time-to-solution.}
\emph{Time-to-solution (TTS)} is a standard performance metric for randomized optimization algorithms~\cite{ronnow2014defining}.
Suppose that a randomized algorithm solves a problem instance with success probability $p$ and average runtime $t$ (over its internal randomness), its TTS with success probability $p_0$ is defined as
\begin{align}\label{eqn:tts}
    \ttspar{p_0} := Rt,\quad \text{where}\ R := \left\lceil \frac{\ln(1 - p_0)}{\ln(1-p)} \right\rceil.
\end{align}
Intuitively, $\ttspar{p_0}$ is the expected time spent for $R$ consecutive runs of the algorithm on a problem instance to guarantee that the overall success probability is at least $p_0$.
In the literature, TTS is often evaluated for a fixed problem instance. In our case, however, the optimization instances are also randomly generated. To reflect the average performance of algorithms when applied to a class of randomly generated optimization instances, we introduce the notion of \emph{average TTS}. Suppose that $\mathcal{A}(\theta)$ is an algorithm (parameterized by some adjustable parameter $\theta$) and $f$ are problem instances randomly drawn from a set $\mathscr{C}$. We define the success probability $p(\theta)$ and average runtime $t(\theta)$ as follows:
\begin{align}
    p(\theta) &\coloneqq \Pr[\text{$\mathcal{A}(\theta)$ solves a problem instance $f$}],\\
    t(\theta) &\coloneqq \mathbb{E}[\text{the time required by $\mathcal{A}(\theta)$ to halt}].
\end{align}
Similar to the standard definition, we define the average TTS of $\mathcal{A}(\theta)$:
\begin{align}
    \ttspar{p_0}(\theta) \coloneqq \left\lceil\frac{\ln(1 - p_0)}{\ln(1-p(\theta))} \right\rceil \cdot t(\theta).
\end{align}
In practice, $p(\theta)$ can be estimated by counting the successful events out of a large number of trials, each with a different optimization instance. The average runtime $t$ will be estimated similarly.
In what follows, for simplicity, we will refer to ``average TTS'' as TTS if not stated otherwise.

\paragraph{TTS for Algorithms with Adjustable Parameters.}
For an optimization algorithm $\mathcal{A}(\theta)$, we can fine-tune the adjustable parameter $\theta$ (e.g., number of iterations, learning rate, etc.) to yield better performance. It is therefore challenging to justify which adjustable parameter $\theta$ we should choose in the estimation of TTS.
To address this issue, we adopt the same approach proposed by R{\o}nnow et al.~\cite{ronnow2014defining} in their study comparing quantum and classical algorithms.
We regard an algorithm parameterized by $\theta$ as an algorithm family and define its TTS (for a given dimensionality $d$) by the minimal TTS in the family over all possible $\theta$.
Admittedly, we can only try a finite number of $\theta$ in reality.
Therefore, we pick the parameter range such that it is more likely to cover the optimal choice(s).

\paragraph{Parameter setup.}
We conduct all of our experiments below on the Zaratan cluster~\cite{zaratan} if not stated otherwise.
For a fixed dimensionality $d$ and parameter $\theta$, we run an optimization algorithm multiple times on instances $F_U$ with Haar-random $U$.
We regard an algorithm successfully solves a problem instance ``minimize $F_U(\xx)$'' if it returns a solution $\vect{y}$ such that $F_U(\vect{y}) \le C/2$, where $C\approx 0.088$ is the (globally) second lowest local minima of $F_U(\xx)$. The success probability (i.e., $p$) of an algorithm $\mathcal{A}(\theta)$ is evaluated as the fraction of successful events out of a large number of tested instances drawn from $\mathcal{C}(w,d)$.
We will not report corresponding TTS if less than 5 success runs are detected since in this case the corresponding estimation of $p$ is likely unreliable.
Each run is single-threaded and assigned one CPU core.
Runtime for estimating $t$ is measured by CPU core time used but not wall time elapsed.
We choose the overall success probability threshold $p_0 = 0.99$.

\subsection{Annealing-based and Monte-Carlo algorithms} \label{sec:dabh}
We tested two stochastic optimization algorithms, dual annealing~\cite{xiang1997generalized} and basin-hopping~\cite{wales1997global}, using their SciPy~\cite{2020SciPy-NMeth} implementation in Python.
An important adjustable parameter $\theta$ of these algorithms is the maximal iteration number. The range of $\theta$, together with other important problem parameters, is presented in the following table.

\begin{align*}
    \begin{array}{c|ccc}
        & \text{Range of }d & \text{Range of }\theta & \text{Number of Runs}\footnotemark  \\
        \hline
        \text{Dual Annealing} & \{1,2,\dots,15\} & \{2^6,2^7,\dots,2^{14}\} & 10000 \\
        \text{Basin-Hopping} & \{1,2,\dots,15\} & \{2^3,2^4,\dots,2^{11}\} & 10000
    \end{array}
\end{align*}
\footnotetext{The number of runs for \emph{each} possible combination of experiment parameters. \label{ftn:nor}}

\fig{dabh} shows the scaling of two algorithms.
Each curve in a plot represents the TTS (parameterized by a specific $\theta$) as a function of the dimension$d$.
Note that the ``true'' TTS of an algorithm is the lower envelope of all the TTS curves, each parameterized by a different $\theta$.
It is clear that TTS for both algorithms scale exponentially in dimensionality $d$.

\begin{figure}[htbp]
\centering
\begin{subfigure}{.5\textwidth}
  \centering
  \includegraphics[width=\textwidth]{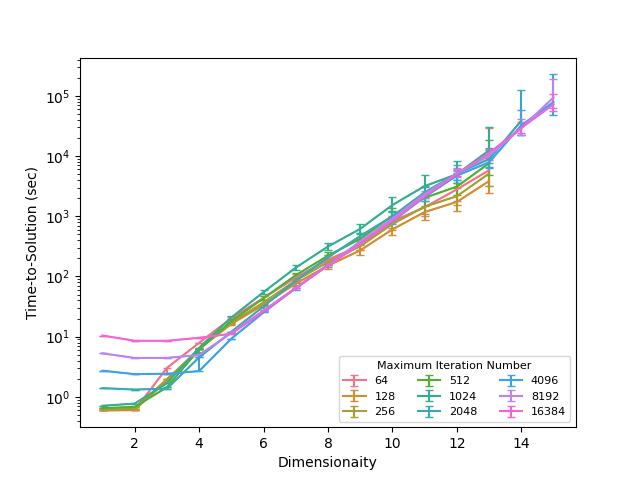}
  \caption{Dual annealing}
\end{subfigure}%
\begin{subfigure}{.5\textwidth}
  \centering
  \includegraphics[width=\textwidth]{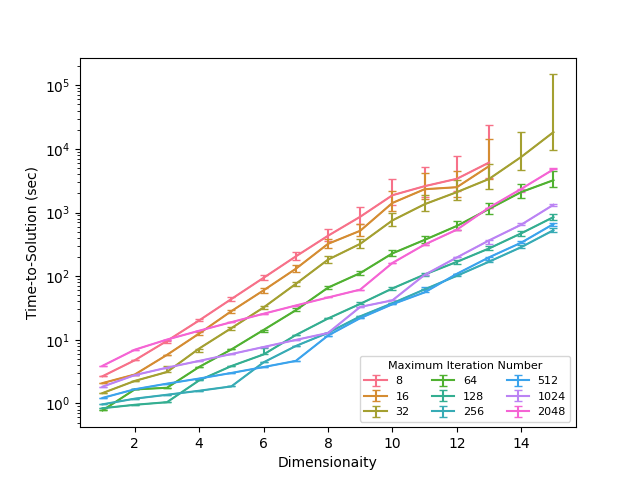}
  \caption{Basin-hopping}
\end{subfigure}
\caption{\small TTS scaling of dual annealing and basin-hopping. Confidence intervals are calculated at the 95\% level. Results suggest that TTS for both algorithms (lower envelopes of curves) scale exponentially in dimensionality $d$.}
\label{fig:dabh}
\end{figure}

\subsection{Newton and quasi-Newton methods}
We also tested two optimization algorithms based on Newton and quasi-Newton methods, Ipopt~\cite{wachter2006implementation} and SQP (sequential quadratic programming)~\cite{nocedal1999numerical} in its SciPy~\cite{2020SciPy-NMeth} implementation.
The backend linear solver for Ipopt is MA27 from Coin-HSL Archive.
The Ipopt experiment was conducted on a consumer laptop (Intel Core i7-8750H CPU 2.20GHz) due to compatibility issues. We do not find an adjustable parameter in these algorithms. Some important experiment parameters are listed below.

\begin{align*}
    \begin{array}{c|cc}
        & \text{Range of }d & \text{Number of Runs}\footref{ftn:nor} \\
        \hline
        \text{Ipopt} & \{1,2,\dots,10\}  & 12000 \\
        \text{SQP} & \{1,2,\dots,15\}  & 100000
    \end{array}
\end{align*}
\fig{local} shows the scaling of the TTS of the two algorithms.
Both scale exponentially in dimensionality $d$. 

\begin{figure}[htbp]
\centering
\begin{subfigure}{.5\textwidth}
  \centering
  \includegraphics[width=\textwidth]{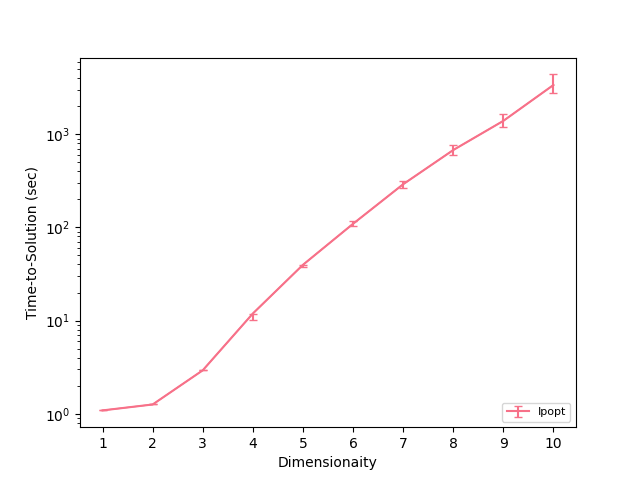}
  \caption{Ipopt}
\end{subfigure}%
\begin{subfigure}{.5\textwidth}
  \centering
  \includegraphics[width=\textwidth]{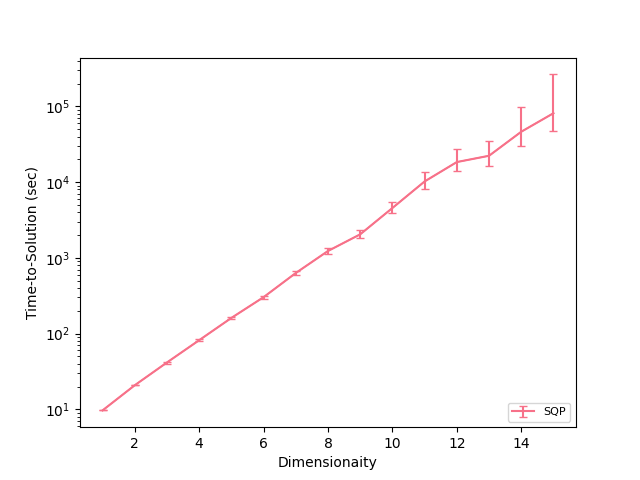}
  \caption{SQP}
\end{subfigure}
\caption{\small TTS scaling of Ipopt and SQP. Confidence intervals are calculated at the 95\% level. Both scale exponentially in dimensionality $d$.}
\label{fig:local}
\end{figure}

\subsection{Stochastic Gradient Descent}
Note that QHD is closely related to classical gradient descent~\cite{leng2023quantum}.
It is thus natural to ask how gradient descent and its variants perform in our instances.
Nevertheless, the fact that $F_U(\xx)$ has $2^d$ local minima immediately rules out the possibility that deterministic gradient-based algorithms have good performance.
Therefore, we focus on stochastic gradient descent (SGD) in this section.

When applying SGD, it is crucial to choose a correct learning rate schedule as it directly determines the performance of SGD.
Although we have some crude empirical rules (e.g., the learning rate should be decaying with time), there is little prior knowledge that we can leverage to choose good learning rate schedules for our instances.
As a result, we choose two types of schedules of particular interest: (1) constant learning rate; and (2) the learning rate schedules that correspond to the time-dependent functions in QHD.
Below we will explain how one can connect the learning rate schedule in SGD with the time-dependent functions in QHD.

{\LinesNumberedHidden
\begin{algorithm}[htbp]
    \caption{Stochastic Gradient Descent}
    \label{algo:sgd}
    \begin{tabular}{lll}
        \textbf{Input:} & $F_U$ & minimization objective \\
        & $T$ & total evolution time \\
        & $s: [0,T] \to \mathbb{R}_+$ & effective learning rate schedule \\
        & $s_{\rm max}$ & maximal learning rate \\
        & $k \in \mathbb{Z}_+ \cup \{\infty\}$ & maximal number of rounds with no smaller objective observed
    \end{tabular}
    
    $t \gets 0$; \tcp{current time}
    Draw $x \in \mathbb{R}^d$ from $\mathcal{N}(\vect{0},d^2)$; \tcp{initial guess} 
    \While {$t < T$}{
        $s_{\rm cur} \gets s(t)$; \tcp{current effective learning rate}
        \If {$s_{\rm cur} < s_{\rm max}$}{
            $x \gets x - s_{\rm cur} (\nabla F_U(x) + \xi)$ with $\xi \sim \mathcal{N}(\vect{0},1)$ \tcp{descent}
            $t \gets t + s_{\rm cur}$
        }\Else{
            $x \gets x - s_{\rm max} (\nabla F_U(x) + \xi)$ with $\xi \sim \mathcal{N}(\vect{0}, s_{\rm cur}/s_{\rm max})$ \tcp{noise scaled}
            $t \gets t + s_{\rm max}$
        }
        update minimal $F_U(x)$ observed, and \\
        \textbf{break} if it is not updated for consecutive $k$ rounds\;
    }
    \textbf{output} minimal $F_U(x)$ observed with corresponding $x$\;
\end{algorithm}
}

In \algo{sgd}, we give the specific version of SGD we used in our experiment. It is worth noting that there are two major differences between \algo{sgd} and vanilla SGD: (1) the \emph{effective} learning rate $s_{\rm cur}$ may not be the actual learning rate used if it exceeds $s_{\rm max}$, and (2) the learning rate schedule is a function of an abstract time $t$ rather than simply a round number.
The reason we made these modifications is that by doing so, our SGD will correspond to the following SDE (stochastic differential equation)~\cite{shi2020learning}:
\begin{align}
    \d x_t  = - \nabla f(x_t) \d t + \sqrt{s(t)} \d W_t, \label{eqn:sgdsde}
\end{align}
where $W_t$ is a standard Brownian motion.
A recent work by Liu, Su, and Li~\cite{liu2023quantum} generalizes this idea and articulates the connection between stochastic differential equations and the Schr\"{o}dinger Equation. They find that the role of the learning rate $s$ in \eqn{sgdsde} is comparable with the \emph{quantum learning rate} $h$ in the Schr\"odinger equation,
\begin{align} \label{eqn:qhd_trans}
    i \frac{\d}{\d t} \Psi = \left(h^2 \left(-\frac12 \nabla^2 \right) + f(x) \right) \Psi.
\end{align}
If we allow the parameter $h$ to be time-dependent, \eqn{qhd_trans} turns out to be nothing but QHD in disguise: 
let $t^\star := 1/2\lambda(t)$ in \eqn{qhd_params} from \algo{qhd}.
\eqn{qhd_params} becomes
\begin{align}
    i \epsilon \frac{\d}{\d t^\star}\ket{\Psi^\epsilon(t^\star)} = H(t^\star) \ket{\Psi^\epsilon(t^\star)}, \quad \text{where }  H(t^\star) =  \frac{1}{(2{t^\star})^2} \left(-\frac12 \nabla^2 \right) + f(x).
\end{align}
And the range of $t^\star$ is $[1/2\lambda_f,\lambda_f/2]$.
The quantum learning rate is then $s = 1/2t^{\star}$.
Translating back to SGD we get the following learning rate schedule $s_{\lambda_f} \colon [0,\lambda_f/2 - 1/2\lambda_f] \to \mathbb{R}_+$,
\begin{align*}
   s_{\lambda_f}(t) = \frac{1}{2 \left( t + 1/2\lambda_f \right)}.
\end{align*}

We run the SGD experiments using the following parameters.
\begin{align*}
    \begin{array}{c|ccccc}
        & \text{Range of }d & \text{Choices of }s & s_{\rm max} & k & \text{Number of Runs}\footref{ftn:nor} \\
        \hline
        \text{SGD} & \{1,2,\dots,15\} & \{2^{-5},2^{-4},\dots,2^{3}\} & 0.01 & 1000 & 10000 \\
        \text{SGD} & \{1,2,\dots,15\} & s_{\lambda_f} \text{ with } \lambda_f\in \{2^1,2^2,\dots,2^{9}\} & 0.01 & \infty  & 10000
    \end{array}
\end{align*}

\fig{sgd} shows the scaling of SGD for two types of learning rate schedules. Here, we regard the learning rate $s$ as a major adjustable parameter that affects the performance of the algorithm. 
In each subfigure, a curve represents the TTS scaling for a learning rate $s$. Again, the true TTS scaling of SGD should be the lower envelope of all curves parameterized by various $s$. Clearly, the TTS of SGD scales exponentially in dimensionality $d$ in both settings.

\begin{figure}[htbp]
\centering
\begin{subfigure}{.5\textwidth}
  \centering
  \includegraphics[width=\textwidth]{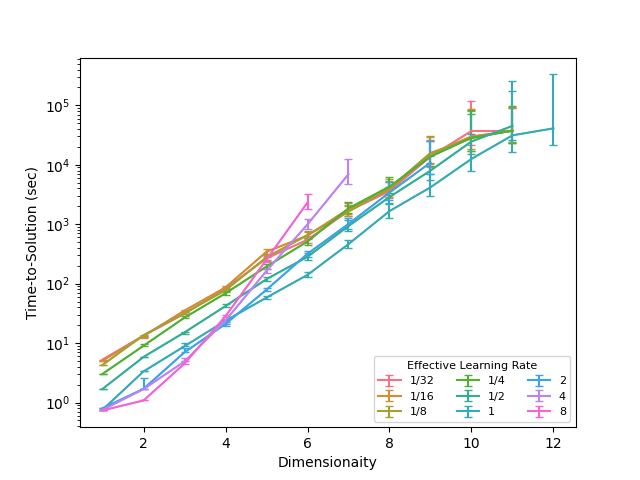}
  \caption{Fixed learning rate}
\end{subfigure}%
\begin{subfigure}{.5\textwidth}
  \centering
  \includegraphics[width=\textwidth]{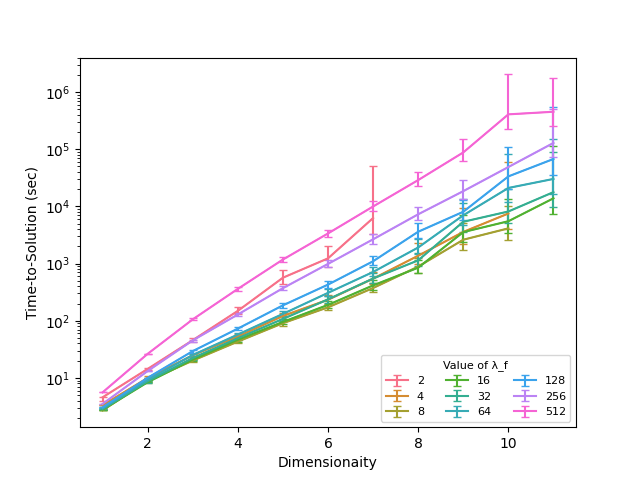}
  \caption{QHD-type learning rate}
\end{subfigure}
\caption{\small TTS scaling of SGD under two types of learning rate schedules respectively. Confidence intervals are calculated at the 95\% level. It is clear that TTS in both settings (lower envelopes of curves) scale exponentially in dimensionality $d$. In some curves, the data points are missing because we do not report the TTS if less than 5 success runs are detected.}
\label{fig:sgd}
\end{figure}

\subsection{Gurobi}

We also conduct an experiment to test the performance of Gurobi~\cite{gurobi}, an industrial-level optimization software based on the branch-and-bound algorithm.
Note that our optimization instances are degree-4 polynomials, while Gurobi only accepts quadratic objectives. The standard resolution is to introduce auxiliary variables to represent higher-degree polynomial terms. For example, the following quartic programming problem,
\begin{align}
    \text{minimize} & \quad x_1x_2x_3x_4, \\
    \text{where} & \quad \vect{x} \in \R^4,
\end{align}
is equivalent to a Quadratically Constrained Quadratic Program (QCQP),
\begin{align}
    \text{minimize} & \quad X_{12}X_{34}, \\
    \text{where} & \quad X_{12} = x_1 x_2, X_{34} = x_3 x_4.
\end{align}
Similarly, to reformulate the degree-4 nonconvex optimization instance as a QCQP, we need to introduce $\Theta(d^2)$ auxiliary variables $X_{ij}$ for $i,j \in \{1,\dots,d\}$ and we add corresponding constraints $X_{ij} = x_i x_j$.

\begin{figure}[htbp]
\centering
\includegraphics[width=.5\textwidth]{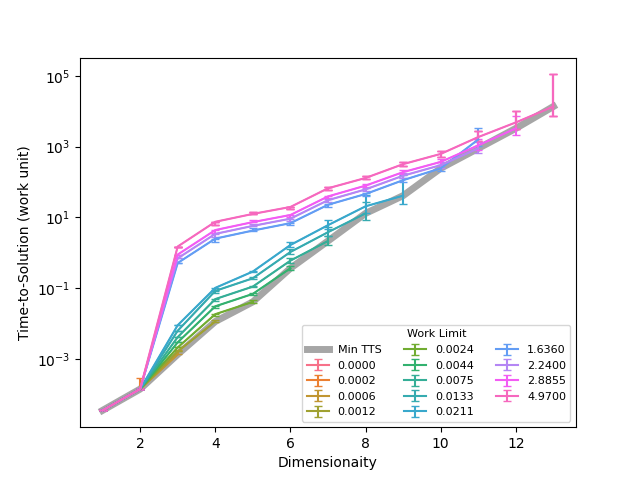}
\caption{\small TTS scaling of Gurobi which is roughly exponential in dimensionality $d$.
Confidence intervals are calculated at the 95\% level. 1 work unit is proportional to 1 second of CPU time. In some curves, the data points are missing because we do not report the TTS if less than 5 success runs are detected.}
\label{fig:gb}
\end{figure}

In our experiment, the adjustable parameter $\theta$ of Gurobi is the \emph{work limit}, which limits the total \emph{work} expended in a run.
The \emph{work} metric, introduced by Gurobi, is proportional to the CPU time spent. However, the pre-factor significantly depends on the hardware and the model (problem instance type) that is being solved. In our setting, a rough estimate suggests that 1 work unit corresponds to $10$--$100$ seconds of CPU time.

Besides the \emph{work limit}, we also set a maximal timeout of 30 minutes for Gurobi. If Gurobi does not solve a problem instance within 30 minutes, we regard it fails for this instance and use the current work limit as the runtime.
Below, we list our experiment parameters.
\begin{align*}
    \begin{array}{c|ccc}
        & \text{Range of }d & \text{Time Limit} & \text{Number of Runs}\footref{ftn:nor} \\
        \hline
        \text{Gurobi} & \{1,2,\dots,15\} & 1800\text{ sec}  & 3000 
    \end{array}
\end{align*}

\fig{gb} shows the TTS scaling of Gurobi. Each TTS curve is measured for a Gurobi program with a fixed $\theta$. In total, 10 TTS curves are depicted. Each curve uses a work limit parameter $\theta$ that yields the optimal performance in at least one dimension $d\in \{1, \dots, 13\}$. Some data points are missing in the plot as we do not report the TTS with an extremely small success probability (i.e., number of successful events less than $5$). The lower envelope of these TTS curves (see the gray curve) gives an estimated scaling of the runtime of Gurobi. The TTS of Gurobi essentially scales exponentially in dimensionality $d$.

\newpage
\bibliographystyle{myhamsplain}
\bibliography{ref}

\appendix
\noindent{\huge \textbf{Appendices}}

\section{Quantum simulation with analytic initial data}\label{append:simulation}
In this section, all notations are the same as in \sec{q_simulation}.

\begin{lemma}
    Let $\Psi(x,t)$ denote the exact solution of \eqn{schrodinger} and $\widetilde{\Psi}(x,t)$ denote the approximated solution by the Fourier spectral method (truncated up to frequency $n$).\footnote{More details on the Fourier spectral method can be found in Section 2.2 (in particular Lemma 1) in \cite{childs2022quantum}.}
    We assume that the initial data $\Psi_0(x)$ is periodic and analytic in $x\in \Omega$. Then, for any integer $n \ge 1$, the error from the Fourier spectral method satisfies
    \begin{align}
        \max_{x, t} |\Psi(x,t) - \widetilde{\Psi}(x,t)| \le 2r^{n/2+1},
    \end{align}
    where $0 < r < \min(\frac{1}{2},\frac{1}{At})$, $A$ is an absolute constant that only depends on $\Psi_0(x)$.
\end{lemma}
\begin{proof}
    We give the proof in one dimension, as the same argument is readily generalized to arbitrary finite dimensions. We assume the initial data $\Psi_0(x)$ is periodic over $[0, 2\pi]$. The analyticity implies that, for any $x \in [0, 2\pi]$, there is a constant $C$ such that
    \begin{align}
        \left|\Psi^{(k)}_0(x)\right| \le C^{k+1}(k!).
    \end{align}
    Therefore, the function $\Psi_0(x)$ admits an analytic continuation in the strip 
    $$\Gamma_0 = \left\{z\in \C: |z - x| < \frac{1}{C}, x \in [0, 2\pi]\right\}.$$
    Due to \cite[Proposition 1]{bourgain1999growth}, there is an absolute constant $A$ such that for any $t \in [0,T]$, 
    \begin{align}\label{eqn:sobolev}
        \left\|\Psi^{(k)}(\cdot,t)\right\| \le A t \left\|\Psi^{(k)}_0\right\| \le A C^{k+1}t (k!),
    \end{align}
    which implies that the wave function $\Psi(x,t)$ is periodic and analytic for any finite $t$. The strip on which $\Psi(x,t)$ admits an analytic continuation is
    $$\Gamma_t = \left\{z\in \C: |z - x| < \frac{1}{ACt}, x \in [0, 2\pi]\right\}.$$
    
    Suppose that the function $\Psi(x,t)$ allows an exact, infinite trigonometric polynomial representation (see \cite[Lemma 16]{childs2022quantum}),
    \begin{align}
        \Psi(x,t) = \sum^\infty_{k=0}c_k(t) e^{ikx}.
    \end{align}
    Let $\Tilde{\Psi}(t,x)$ be the truncated Fourier series up to $k = n/2$, then the error from the Fourier spectral method satisfies
    \begin{align}
        |\Psi(x,t) - \Tilde{\Psi}(x,t)| \le \sum^\infty_{k=n/2+1}|c_k(t)|.
    \end{align}
    Meanwhile, the function $\Psi(x,t)$ has an analytic continuation defined by $\Psi(x,t) = \sum^\infty_{k=0} c_k(t) z^k$
    for $z \in \Gamma_t$. By Cauchy's integral formula, for any simply connected curve $\gamma$ on the strip $\Gamma_t$, we have that
    \begin{align}
        c_k(t) = \frac{1}{k!}\int_\gamma \Psi^{(k)}(z,t) \d z. 
    \end{align}
    Together with \eqn{sobolev}, it turns out that $|c_k(t)| \le r^{k}$ for some $0 < r < \frac{1}{At}$. Therefore, the error from the Fourier spectral method satisfies
    \begin{align}
        |\Psi(x,t) - \Tilde{\Psi}(x,t)| \le \sum^\infty_{k=n/2+1}|c_k(t)| \le \frac{r^{n/2+1}}{1-r}.
    \end{align}
    Moreover, if we force $r < 1/2$, the error is bounded by $2r^{n/2+1}$.
\end{proof}

To ensure that the simulation error is bounded by $\eta$, we may choose the truncation number 
$$n = \left\lceil 2\left(\frac{\log(4/\eta)}{\log(1/r)} - 1\right) \right\rceil \le \bigO\left(\log(1/\eta)\right).$$
Next, by plugging this truncation number in Equation (113) in \cite{childs2022quantum}, we prove \thm{spectral-method}.

\section{Sub-Gaussian ground states}
\subsection{Probability toolbox}
\begin{lemma}\label{lem:subgaussian_convert}
    Let $X\in \R$ be a random variable such that $\mathbb{E}[e^{X^2/K}] \le 2$, where $K>0$ is a real number. 
    Then, for any $s\in \R$, we have
    \begin{align}
        \mathbb{E}[e^{sX}] \le e^{2s^2K}.
    \end{align}
\end{lemma}
\begin{proof}
    This is a standard result of sub-Gaussian random variables. Here, we give the proof for the completeness of this paper.
    By Markov's inequality, we have
    \begin{align*}
        \Pr[|X| > t] = \Pr\left[e^{X^2/K} > e^{t^2/K}\right] \le \frac{\mathbb{E}[e^{X^2/K}]}{e^{t^2/K}} \le 2 e^{-t^2/K}.
    \end{align*}
    With the above inequality, we can estimate the $k$-th moment of $X$. For any positive integer $k \ge 1$,
    \begin{align*}
        \mathbb{E}[|X|^k] &= \int^\infty_0 \Pr[|X|^k > t]~\d t = \int^\infty_0 \Pr[|X| > t^{1/k}]~\d t\\
        &\le 2 \int^\infty_0 e^{-t^{2/k}/K}~\d t = K^{k/2}k\Gamma(k/2),
    \end{align*}
    where $\Gamma(z)\coloneqq \int^\infty_{0}t^{z-1}e^{-t}~\d t$ is the Gamma function. 
    
    For any $s \in \R$, we use the Taylor expansion of the exponential function $e^{sX}$ and apply the dominated convergence theorem,
    \begin{align*}
        \mathbb{E}[e^{sX}] &\le 1 + \sum^\infty_{k=1} \frac{s^k \mathbb{E}[|X|^k]}{k!}
        \le 1 + \sum_k \frac{(s^2K)^{k/2}k\Gamma(k/2)}{k!}\\
        &= 1 + \sum^\infty_{k=1}\frac{(s^2K)^k 2k \Gamma(k)}{(2k)!} + \sum^\infty_{k=1}\frac{(s^2K)^{k+1/2} (2k+1) \Gamma(k+1/2)}{(2k+1)!}\\
        &\le 1 + \left(2 + \sqrt{s^2K}\right)\sum^\infty_{k=1}\frac{(s^2K)^k k!}{(2k)!}
        \le 1 + \left(1 + \sqrt{\frac{s^2K}{4}}\right)\sum^\infty_{k=1}\frac{(s^2K)^k }{k!} \le e^{2s^2K}.
    \end{align*}
    Note that we use the inequality $2(k!)^2 \le (2k)!$ in the second-to-last step.
\end{proof}

\begin{lemma}\label{lem:subgaussian_rescale}
    Let $X\in \R$ be a random variable such that $\mathbb{E}[e^{X^2/K}] \le \alpha$ with $\alpha > 1$. Then, we have
    \begin{align}
        \mathbb{E}[e^{X^2/(\alpha-1)K}] \le 2.
    \end{align}
\end{lemma}
\begin{proof}
    Let $r > 1$ be a real-valued parameter. Using the Fubini theorem,
    \begin{align*}
        \mathbb{E}[e^{X^2/rK}] &= 1 + \sum^\infty_{k=1} \frac{\mathbb{E}[(X^2/rK)^k]}{k!}
        \le 1 + \frac{1}{r}\sum^\infty_{k=1} \frac{\mathbb{E}[(X^2/K)^k]}{k!}
        \le 1 + \frac{\alpha - 1}{r}. 
    \end{align*}
    Therefore, if we choose $r = \alpha - 1$, we end up with $\mathbb{E}[e^{X^2/(\alpha-1)K}] \le 2$.
\end{proof}

\begin{theorem}[Theorem 1, \cite{hsu2012tail}]\label{thm:subgaussian_tail}
    Let $A \in \R^{m\times n}$ be a matrix, and let $\Sigma = A^T A$. Suppose that $x = (x_1,\dots,x_n)$ is a random vector such that, for some $\sigma \ge 0$,
    \begin{align}
        \mathbb{E}\left[e^{r^Tx}\right] \le e^{\|r\|^2\sigma^2/2}
    \end{align}
    for all $r \in \R^d$. Then, for all $c > 0$, 
    \begin{align}
        \Pr\left[\|Ax\|^2 > \sigma^2\left(\Tr(\Sigma) + 2\sqrt{\Tr(\Sigma^2)c} + 2\|\Sigma\|c\right)\right] \le e^{-c}.
    \end{align}
\end{theorem}

\subsection{Proof of \lem{subgaussian}}\label{append:proof_subgaussian}
\begin{proof}
    Note that the ground state of the Hamiltonian $\hat{H}(\lambda)$ is exactly the same as $$2\lambda\hat{H}(\lambda) = - \nabla^2 + 2 \lambda^2 w(x),$$
    where the effective potential function is $\lambda^2 w(x)$.
    By \cite[Theorem 3.4]{hislop2012introduction}, there is a constant $\alpha > 0$ such that
    \begin{align}\label{eqn:decay}
        \int_\R e^{d_0(x)}|\phi_\lambda(x)|^2 ~\d x \le \alpha,
    \end{align}
    where $d_0(x)$ is the Agmon distance between $x$ and $0$ (see \cite[Definition 3.2]{hislop2012introduction}). It is well known that the Agmon distance scales as the square root of the potential function, i.e.,
    \begin{align}
        d_0(x) \sim (\lambda^2 w(x))^{1/2} \text{ as } |x|\to \infty.
    \end{align}
    By \defn{1}, the function $w(x)$ grows as fast as a quartic function, which means there exists a positive constant $c'$ such that $w(x) \ge c'(x-x^*)^4$. It turns out that
    $$d_0(x) \ge c'' \lambda(x - x^*)^2,$$
    where $c''$ is an absolute constant.
    It follows from \eqn{decay} that
    \begin{align}\label{eqn:estimate1}
        \int_\R e^{c''\lambda(x-x^*)^2}|\phi_\lambda(x)|^2 ~\d x \le \alpha.
    \end{align}
    Let $X \sim |\phi_\lambda (x)|^2$ and $\xi \coloneqq X - x^*$, we rewrite \eqn{estimate1} as $\mathbb{E}[e^{c''\lambda \xi^2}] \le \alpha$. Without loss of generality, we assume $\alpha > 1$. Then, \lem{subgaussian_rescale} implies that $\mathbb{E}[e^{c''\lambda \xi^2/(\alpha - 1)}] \le 2$. Now, we invoke \lem{subgaussian_convert} to obtain that
    \begin{align}
        \mathbb{E}[e^{s\xi}] \le e^{s^2\sigma^2/2},
    \end{align}
    where $\sigma^2 = \beta/\lambda$ with $\beta = 4(\alpha-1)/c''$. Note that $\beta$ only depends on our choice of $w$.
\end{proof}

\section{Auxiliary lemmas}\label{append:technical}
\begin{lemma}\label{lem:robust_state}
    Suppose that $\psi_1$, $\psi_2$ are two unit vectors in $\H$ such that
    \begin{align}
        \|\psi_1 - \psi_2\| \le \delta.
    \end{align}
    Let $O$ be a bounded operator on $\H$ such that $\|O\| \le 1$. Then, we have
    \begin{align}
        \left|\bra{\psi_1}O\ket{\psi_1} - \bra{\psi_2}O\ket{\psi_2}\right| \le 2\delta.
    \end{align}
\end{lemma}
\begin{proof}
    By the triangle inequality, 
    \begin{align*}
        \left|\bra{\psi_1}O\ket{\psi_1} - \bra{\psi_2}O\ket{\psi_2}\right| \le \left|(\bra{\psi_1} - \bra{\psi_2})O\ket{\psi_1}\right| + \left|\bra{\psi_2}O(\ket{\psi_1} - \ket{\psi_2})\right| \le 2\delta.
    \end{align*}
\end{proof}

\vspace{4mm}
The following is the proof of \lem{time_dilation}.

\begin{proof}
    To simulate the QHD dynamics as given in \eqn{qhd_params}, we consider the change of variable,
\begin{align}
    \xi = \lambda_f + \frac{1}{\lambda_f} - \frac{1}{\lambda(t)} = \lambda_f\left(1 - e^{-2t}\right) + \frac{1}{\lambda_f}.
\end{align}
We see that $\xi(t)$ is an increasing function in $t$ such that $\xi(0) = \frac{1}{\lambda_f}$, $\xi(t_f) = \lambda_f$. Moreover, we can express $t$ in terms of $\xi$:
\begin{align}
    t(\xi) = -\frac{1}{2}\log\left(-\frac{\xi}{\lambda_f} + 1 + \frac{1}{\lambda^2_f}\right).
\end{align}

Using the change of variable $t\mapsto \xi(t)$ and we define $\ket{\Psi^\epsilon(\xi)} = \ket{\Psi^\epsilon(\xi(t))}$, the Schr\"odinger equation \eqn{qhd_params} becomes
\begin{align}
    i \epsilon \frac{\d}{\d \xi}\ket{\Psi^\epsilon(\xi)} = \frac{1}{2}\left[-\frac{1}{2}\nabla^2 + \varphi(\xi) f\right]\ket{\Psi^\epsilon(\xi)},
\end{align}
where 
\begin{align}
    \varphi(\xi) \coloneqq \frac{1}{\left(-\xi + \lambda_f + 1/\lambda_f\right)^2} \in [1/\lambda^2_f, \lambda^2_f].
\end{align}

We can also absorb the parameter $\epsilon$ into the Hamiltonian by dilating the time scale $s \coloneqq \xi/\epsilon \in J \coloneqq [\frac{1}{\epsilon \lambda_f}, \frac{\lambda_f}{\epsilon}]$. Let $\varphi(s) \coloneqq \varphi(\epsilon s)$, we end up with the effective dynamics described by \eqn{effective_dynamics}.
\end{proof}

\end{document}